\documentclass{article}
\pdfoutput=1 

\usepackage{graphicx}
\usepackage{amsmath}
\usepackage{amsthm}
\usepackage{amssymb}
\usepackage{authblk}
\usepackage{fullpage}

\newtheorem{theorem}{Theorem}
\newtheorem{remark}[theorem]{Remark}
\newtheorem{lemma}[theorem]{Lemma}
\newtheorem{cor}[theorem]{Corollary}
\newcommand{\naturals}{\mathbb N}
\newcommand{\diag}{{\rm{diag}}}

\allowdisplaybreaks[1]

\begin{document}

\title{State Transfer on Paths with Weighted Loops}

\author[1]{Stephen Kirkland%
\thanks{Research supported in part by NSERC Discovery Grant RGPIN--2019--05408.}}
\author[1]{Christopher M. van Bommel%
\thanks{Research supported in part by PIMS.}%
\thanks{Corresponding author.}} 
\affil[1]{Department of Mathematics, University of Manitoba, Winnipeg, MB, Canada.  Stephen.Kirkland@umanitoba.ca,  Christopher.vanBommel@umanitoba.ca}

\maketitle

\begin{abstract}
We consider the fidelity of state transfer on an unweighted path on $n$ vertices, where a loop of weight	$w$ has been appended at each of the end vertices. It is known that if $w$ is transcendental, then there is pretty good state transfer from one end vertex to the other; we prove a companion result to that fact, namely that there is a dense subset of $[1,\infty)$ such that if $w$ is in that subset, pretty good state transfer between end vertices is  impossible. Under mild hypotheses on $w$ and $t$, we derive upper and  lower bounds on the fidelity of state transfer between end vertices at readout time $t$. Using those bounds, we  localise the readout times for which that fidelity  is close to $1$. We also provide expressions for, and bounds on, the sensitivity of the fidelity of state transfer between end vertices, where the sensitivity is with respect to either the readout time or the weight $w$. Throughout, the results rely on detailed knowledge of the eigenvalues and eigenvectors of the associated adjacency matrix.  
\end{abstract}

\section{Introduction}\label{sec:intro} 

An important task in the area of quantum information processing is the faithful transmission of information.  One avenue for the implementation of this task is a network of locally coupled spins, with the information as an excitation in the network, which is initialized at a source and then propagates according to Schr\"odinger's equation:
\[
i h \frac{\partial}{\partial t} \varphi = H \varphi.
\]
We can then consider the \emph{fidelity} of transmission  between a source and a target -- i.e., the probability that an excitation initialized at a source node is found at the target node -- to describe the quality or accuracy of the transmission.  The protocol for quantum communication through unmeasured and unmodulated spin chains was presented by Bose~\cite{B03}, and led to the interpretation of quantum channels implemented by spin chains as wires for transmission of excitations.  Here, we will  consider the Hamiltonian given by the adjacency matrix of the network, though other choices are possible.

If there exists a time when the fidelity of transmission is 1, then  we have \emph{perfect state transfer}.  In general, examples of perfect state transfer are rare, and networks known to exhibit perfect state transfer are highly structured, require specific edge weights, or require external control.  Christandl et al.~\cite{CDEL04, CDDEKL05} demonstrated that perfect state transfer can be achieved on Cartesian powers of the path on two vertices and the path on three vertices.  By taking the quotient of the former, they obtained edge weightings that  admit perfect state transfer on paths of arbitrary length.  Albanese et al.~\cite{ACDE04} extended this result to mirror inversion of an arbitrary quantum state.  Pemberton-Ross and Kay~\cite{PK11} constructed regular networks that permit perfect state transfer between any pair of sites through the use of controlled external magnetic fields.  Karimipour, Sarmadi Rad, and Asoudeh~\cite{KSA12} provided examples requiring less external control.  The study of perfect state transfer has been surveyed by Kay~\cite{K10} from a physics standpoint and Godsil~\cite{G12} from a mathematics perspective.

Christandl et al.~\cite{CDDEKL05} demonstrated that perfect state transfer does not occur between the end vertices of an unweighted path on at least 4 vertices.  Stevanovic~\cite{S11} and Godsil~\cite{G12} independently extended this result to any pair of vertices of a path.  Kempton, Lippner, and Yau~\cite{KLY17a} ruled out perfect state transfer on unweighted paths even in the presence of a fixed external magnetic field.

As a slightly weaker formulation, if there exist times at which fidelities arbitrarily close to 1 are obtained, then  we have \emph{pretty good state transfer} (PGST), or \emph{almost perfect state transfer}.  This less restrictive property recognizes that perfect state transfer, particularly when it is achieved via edge weights or external control, is subject to imprecision in its manufacture or implementation.  Vinet and Zhedanov~\cite{VZ12} demonstrated that sufficient conditions for pretty good state transfer are that the network be mirror-symmetric, and that the eigenvalues of the underlying graph be rationally independent.  Godsil, Kirkland, Severini, and Smith~\cite{GKSS12} demonstrated that pretty good state transfer occurs between the end vertices of an unweighted path of length $n$, where $n + 1$ is a prime, twice a prime, or a power of 2.  Coutinho, Guo, and van Bommel~\cite{CGvB17} identified an infinite family of paths admitting pretty good state transfer between internal vertices, and van Bommel~\cite{vB19} verified that no other path lengths admit pretty good state transfer.  Banchi, Coutinho, Godsil, and Severini~\cite{BCGS17} considered a Hamiltonian formed by the Laplacian, and show that in this case, pretty good state transfer occurs between the end vertices of a path whose length is a power of 2.

In light of the paucity of examples of perfect and pretty good state transfer on unweighted paths, perturbed paths have also been considered to achieve a notion of \emph{asymptotic state transfer}, that is, the fidelity of state transfer approaches 1 as some parameter of the graph changes.  W\'ojcik et al.~\cite{WLKGGB05} demonstrated asymptotic state transfer by changing the weights of the end edges of the path.  Casaccino, Lloyd, Mancini, and Severini~\cite{CLMS09} provided numeric evidence of asymptotic state transfer by adding loops of weight $w$ to the first and last vertices of a path.  Linneweber, Stolze, and Uhrig~\cite{LSU12} confirmed this result analytically.  Lin, Lippner, and Yau~\cite{LLY} considered a Hamiltonian formed by the normalised Laplacian, and show asymptotic state transfer when adding loops to two (or sometimes three) vertices if the vertices are sufficiently symmetric (a notion that is made precise in \cite{LLY}). Lorenzo, Apollaro, Sindona, and Plastina~\cite{LASP13} considered asymptotic state transfer by adding loops of weight $w$ to the second and second-last vertices of a path.  Chen, Mereau, and Feder~\cite{CMF16} obtain asymptotic state transfer by adding weighted edges from the third and third-last vertices of a path.  

In this paper we consider  the  family of graphs $P_n^w$ constructed as follows:  start with an unweighted path on $n$ vertices, then add a loop of weight $w$ at each end vertex. We focus on the fidelity of state transfer from one end vertex to the other. Our work here is motivated in part by a result of Kempton, Lippner, and Yau~\cite{KLY}, which shows that if $w$ transcendental, then there is pretty good state transfer from one end vertex of $P_n^w$ to the other. (We note in passing that \cite{KLY} contains many other results, and considers a broad family of weighted graphs that includes $P_n^w$.) Consequently, there is a dense subset of $\mathbb{R}$ such that if  $w$ is chosen from that subset, then pretty good state transfer between the end vertices is guaranteed. However, as we show in Section \ref{sec:nopgst}, there is also a dense subset of $[1, \infty)$ such that if  $w$ is chosen from that subset, then pretty good state transfer between the end vertices is impossible. Evidently great care and accuracy is needed in choosing $w$ if one is seeking to ensure pretty good state transfer for $P_n^w$. 

In view of this last observation, our objective in the sequel is to estimate the fidelity of state transfer between the end vertices of $P_n^w$, without considering whether or not pretty good state transfer holds for the particular choice of $w$.  Specifically, in Section \ref{subsec:fid} we prove upper and lower bounds on that fidelity in terms of $n, w$ and the readout time. The lower bound furnishes a guaranteed minimum on the fidelity of state transfer at a particular readout time, and for certain choices of the readout time, that lower bound can be made arbitrarily close to $1$ via suitable choice of $w$.  The upper bound, in turn, informs the choice of readout time so as to meet or exceed  a given  threshold for the fidelity. In Sections \ref{sec:time} and \ref{sec:weight} we prove bounds on the sensitivity of the fidelity with respect to the readout time and the value of $w$, respectively. The inequalities suggest that the fidelity is robust with respect to the readout time, but not with respect to the value of $w$. Throughout, our results rely on detailed  information about the eigenvalues and eigenvectors of the adjacency matrix of $P_n^w$; that information is developed in Sections \ref{sec:l1l2} and \ref{sec:other}. While state transfer problems are often treated from an algebraic viewpoint, the results obtained below arise from viewing the properties of state transfer primarily from an analytic perspective.  

Taken together, our results  quantify not only the fidelity of state transfer between the end vertices of $P_n^w$, but also  the choices of readout time that lead to  large fidelity. This is in marked contrast to the notion of pretty good state transfer, which, while  asserting the existence of a sequence of readout times for which the fidelity converges to $1,$  is silent on: i) when those readout times are, and  ii) how close to $1$ the fidelity is at those readout times. We note that several of our results are reminiscent of those of Lin, Lippner, and Yau~\cite{LLY}.

\section{Preliminaries} 

In this section we present some technical material that will assist us in deriving our later results. 

 The path $P_n^w$ is constructed from the unweighted path on $n$ vertices, with loops of weight $w$ added to each end vertex.  We assume that $V(P_n^w) = \{1, 2, \ldots, n\}$ and $E(P_n^w) = \{ \{1, 2\}, \{2, 3\}, \ldots, \{n - 1, n\} \}$.  In general, we assume that $w > 1$.  For a given $n$, we let $A(w)$ denote the adjacency matrix of $P_n^w$.  We consider the fidelity of state transfer from one end of the path to the other with respect to uniform XX couplings (represented by the edges of the path) and potentials (represented by loops), whose Hamiltonian is given by
\[
H_{XX} = \frac{1}{2} \sum_{j = 1}^{n - 1} (X_{j} X_{j + 1} + Y_j Y_{j + 1}) + w (Z_1 + Z_n),
\]
where $X_j$, $Y_j$, and $Z_j$ are the standard Pauli operators on site $j$.  Restricting to the single-excitation subspace allows us to take $A(w)$ as the Hamiltonian; then the solution to Schr\"odinger's equation is given by $\varphi(t) = U(t) \varphi(0)$, $U(t) = \exp (i t A(w))$ (where we incorporate Planck's constant $h$ in the time interval $t$).  Using the  spectral decomposition, we calculate the matrix exponential by
\[
U(t) = \sum_\lambda \exp(i t \lambda) E_\lambda,
\]
where the sum is taken over the eigenvalues $\lambda$ of $A(w)$ and $E_\lambda$ is the projection onto the $\lambda$-eigenspace.  We will denote the eigenvalues of $A(w)$ by $\lambda_1 \ge \lambda_2 \ge \cdots \ge \lambda_n$.  We let $V = [v_{ij}]$ be an orthogonal matrix that diagonalises $A(w)$, i.e.\ $V^T A(w) V = \diag(\lambda_1, \ldots, \lambda_n)$, so the columns of $V$ give the eigenvectors of $A(w)$.

Formally, we say we have \emph{perfect state transfer} from vertex $a$ to vertex $b$ at time $\hat t$ if $|U(\hat t)_{a, b}| = 1$ and \emph{pretty good state transfer} from vertex $a$ to vertex $b$ if, for every $\epsilon > 0$, there exists a time $t_\epsilon$ such that $|U(t_\epsilon)_{a, b}| > 1 - \epsilon$.  In general, we represent the fidelity of transfer from $a$ to $b$ at time $t$ by $p(t) = |U(t)_{a, b}|^2$.

The following lemma characterising pretty good state transfer is originally due to Banchi, Coutinho, Godsil, and Severini~\cite{BCGS17}; the form below is due to Kempton, Lippner, and Yau~\cite{KLY}.

\begin{lemma} \cite{BCGS17, KLY} \label{lem:pgst}
Let $u, v$ be vertices of $G$, and $H$ the Hamiltonian.  Then pretty good state transfer from $u$ to $v$ occurs at some time if and only if the following two conditions are satisfied:
\begin{enumerate}
    \item Every eigenvector $x$ of $H$ satisfies either $x_u = x_v$ or $x_u = - x_v$.
    \item Let $\{\lambda_i\}$ be the eigenvalues of $H$ corresponding to eigenvectors with $x_u = x_v \neq 0$, and $\{\mu_j\}$ the eigenvalues for eigenvectors with $x_u = - x_v \neq 0$.  Then if there exist integers $\ell_i$, $m_j$ such that if
    \begin{align*}
    \sum_i \ell_i \lambda_i + \sum_j m_j \mu_j &= 0 \\
    \sum_i \ell_i + \sum_j m_j &= 0
    \end{align*}
    then
    \[
    \sum_i m_i \text{ is even.}
    \]
\end{enumerate}
\end{lemma}

 Throughout, we let $e_j$ denote the standard basis vector which consists of a 1 in the $j$th row and every other entry is 0.  For $n$ even, we define the matrices 
\begin{equation}\label{eq:b1b2} 
B_1(w)=\left[\begin{array}{cccccc} 
w&1&0&0&\ldots&0\\
1&0&1&0&\ldots &0\\
0&1&0&1&\ldots&0\\
\vdots&\vdots & \ddots &\ddots&\ddots &\vdots \\
0&0&\ldots&1&0&1\\
0&0&\ldots&0&1&1
\end{array}\right], \quad B_2(w)=\left[\begin{array}{cccccc} 
w&1&0&0&\ldots&0\\
1&0&1&0&\ldots &0\\
0&1&0&1&\ldots&0\\
\vdots&\vdots & \ddots &\ddots&\ddots &\vdots \\
0&0&\ldots&1&0&1\\
0&0&\ldots&0&1&-1
\end{array}\right],
\end{equation} 
where $B_1(w)$ and $B_2(w)$ are of order $\frac{n}{2}$, and for $n$ odd, we define the matrices
\begin{equation}\label{eq:c1c2} 
C_1(w)=\left[\begin{array}{cccccc} 
w&1&0&0&\ldots&0\\
1&0&1&0&\ldots &0\\
0&1&0&1&\ldots&0\\
\vdots&\vdots & \ddots &\ddots&\ddots &\vdots \\
0&0&\ldots&1&0&\sqrt{2}\\
0&0&\ldots&0&\sqrt{2}&0
\end{array}\right], \quad C_2(w)=\left[\begin{array}{cccccc} 
w&1&0&0&\ldots&0\\
1&0&1&0&\ldots &0\\
0&1&0&1&\ldots&0\\
\vdots&\vdots & \ddots &\ddots&\ddots &\vdots \\
0&0&\ldots&1&0&1\\
0&0&\ldots&0&1&0
\end{array}\right],
\end{equation}
  where $C_1(w)$ is of order $\frac{n+1}{2}$ and $C_2(w)$ is of order $\frac{n-1}{2}$.  We will establish by Lemma~\ref{lem:interlace} that the odd-indexed eigenvalues of $A(w)$ are the eigenvalues of $B_1(w)$ or $C_1(w)$ and the corresponding eigenvectors satisfy $v_{1i} = v_{ni} \neq 0$ and the even-indexed eigenvalues of $A(w)$ are the eigenvalues of $B_2(w)$ or $C_2(w)$ and the corresponding eigenvectors $x$ satisfy $v_{1i} = -v_{ni} \neq 0$, allowing us to apply the linear combination condition of Lemma~\ref{lem:pgst} to determine the presence of pretty good state transfer.

We note the following results on the eigenvalues and eigenvectors of $A(w)$; we defer the proofs of these results to Sections~\ref{sec:l1l2} and \ref{sec:other}. We have:
\begin{align}
    \label{eqn:eval} \lambda_j &= 2 \cos (\theta_j) \text{ for some } \theta_j \text{ such that } \frac{(j - 2) \pi}{n - 1} \le \theta_j \le \frac{j \pi}{n + 1}, \quad 3 \le j \le n;  \\
    \label{eqn:v11} v_{11}^2 &= \left(2+ \frac{8}{n-1}\sum_{\ell=1}^{\lceil\frac{n-1}{2}\rceil } \frac{\sin^2\left(\frac{(2\ell-1)\pi}{n-1}\right)}{\left(\lambda_1 - 2\cos\left(\frac{(2\ell-1)\pi}{n-1}\right)\right)^2} \right)^{-1};  \\
\label{eqn:v12} v_{12}^2 &= \left(2+ \frac{8}{n-1}\sum_{\ell=1}^{\lceil\frac{n-2}{2}\rceil } \frac{\sin^2\left(\frac{2\ell\pi}{n-1}\right)}{\left(\lambda_2 - 2\cos\left(\frac{2\ell\pi}{n-1}\right)\right)^2} \right)^{-1}; \\
\label{eqn:v1j}    v_{1j}^2 &= \begin{cases}
\frac{\displaystyle 2 \cos^2 \left( \frac{n - 1}{2} \theta_j \right)}{\displaystyle n + \frac{\sin (n \theta_j)}{\sin (\theta_j)}}, & j \text{ odd}, \\ ~ \\
\frac{\displaystyle 2 \sin^2 \left( \frac{n - 1}{2} \theta_j \right)}{\displaystyle n - \frac{\sin (n \theta_j)}{\sin (\theta_j)}}, & j \text{ even},
\end{cases} \quad \text{where } \lambda_j = 2 \cos (\theta_j). 
\end{align}

\section{No PGST for a Dense Subset of Values for $w$}\label{sec:nopgst} 

As observed in Section \ref{sec:intro}, if $w$ is transcendental, then there is PGST between the end vertices of $P_n^w$. Our goal in this section is to establish the existence of a dense subset of values for $w$ that prevents PGST between end vertices. The following technical lemmas will help to accomplish that objective. 

\begin{lemma} \label{lem:gap23}
For $w \ge 1$, $v_{12}^2(w) - v_{13}^2(w) > 0$.
\end{lemma}

\begin{proof}
We will establish that $\lambda_2(w) - \lambda_3(w) \ge \lambda_2(1) - \lambda_3(1)$ and $v_{12}^2(w) - v_{13}^2(w) > 0$ for all $w \ge 1$.  For $w = 1$, we have $\lambda_2(w) - \lambda_3(w) = \lambda_2(1) - \lambda_3(1)$ and  $\lambda_2(1) = 2 \cos \left( \frac{\pi}{n} \right)$ and $\lambda_3(1) = 2 \cos \left( \frac{2 \pi}{n} \right)$.  Applying (\ref{eqn:v1j}), we obtain
\begin{align*}
    &\quad v_{12}^2(1) - v_{13}^2(1) \\
    &= \frac{\displaystyle 2\sin^2 \left(\frac{n - 1}{2} \left(\frac{\pi}{n}\right)\right)}{\displaystyle n - \frac{\sin (n \frac{\pi}{n})}{\sin (\frac{\pi}{n})}} - \frac{\displaystyle 2\cos^2 \left(\frac{n - 1}{2} \left(\frac{2 \pi}{n}\right)\right)}{\displaystyle n + \frac{\sin (n \frac{2 \pi}{n})}{\sin (\frac{2 \pi}{n})}} \\
    &= \frac{2}{n} \left(\sin^2 \left(\frac{(n - 1) \pi}{2n}\right) - \cos^2 \left(\frac{(n - 1) \pi}{n}\right) \right) \\
    &= \frac{-1}{n} \left( \cos \left(\frac{(n - 1) \pi}{n}\right) + \cos \left(\frac{2 (n - 1) \pi}{n}\right)  \right) \\ 
    &= \frac{1}{n} \left(\cos\left(\frac{\pi}{n}\right) - \cos\left(\frac{2\pi}{n}\right)\right) \\
    &> 0.
\end{align*}

Next, suppose for some $w^* \ge 1$ that $\lambda_2(w^*) - \lambda_3(w^*) \ge \lambda_2(1) - \lambda_3(1)$ and $v_{12}^2(w^*) - v_{13}^2(w^*) > 0$.  Since $\frac{d}{dw} (\lambda_2(w) - \lambda_3(w)) = 2 v_{12}^2(w) - 2 v_{13}^2(w) > 0$, then there exists a $\delta > 0$ such that for all $w^* < w \le w^* + \delta$, we have $\lambda_2(w) - \lambda_3(w) \ge \lambda_2(w^*) - \lambda_3(w^*) \ge \lambda_2(1) - \lambda_3(1)$.  Let $\theta$ be such that $\lambda_3(w) = 2 \cos (\theta)$; note that $\frac{\pi}{n - 1} \le \theta < \frac{2\pi}{n}$.  Since $2 \cos (\theta - \frac{\pi}{n}) - 2 \cos \theta < 2 \cos (\frac{\pi}{n}) - 2 \cos (\frac{2 \pi}{n})$, then $\lambda_2(w) \ge 2 \cos (\theta - \frac{\pi}{n})$.

We observe
\begin{align*}
    &\quad \left( n - \frac{\sin \left( n \left(\theta - \frac{\pi}{n} \right) \right)}{\sin \left(\theta - \frac{\pi}{n} \right)} \right) - \left( n + \frac{\sin \left( n \theta \right)}{\sin \left(\theta \right)} \right)
    = \sin (n \theta) \left(\csc \left(\theta - \frac{\pi}{n} \right) - \csc (\theta) \right) 
    < 0,
\end{align*}
as $\sin(n \theta) < 0$ and $\left(\csc \left(\theta - \frac{\pi}{n} \right) - \csc (\theta) \right) > 0$.  Moreover, we have 
\[
2 \sin^2 \left( \frac{n - 1}{2} \left(\theta - \frac{\pi}{n} \right) \right) = 2 \cos^2 \left( \frac{n - 1}{2} \left(\theta + \frac{\pi}{n(n - 1)} \right) \right) > 2 \cos^2 \left( \frac{n - 1}{2} \theta \right),
\]
and therefore
\[
\frac{\displaystyle 2 \sin^2 \left( \frac{n - 1}{2} \left(\theta - \frac{\pi}{n} \right) \right)}{\displaystyle n - \frac{\sin \left( n \left(\theta - \frac{\pi}{n} \right) \right)}{\sin \left(\theta - \frac{\pi}{n} \right)}} > \frac{\displaystyle 2 \cos^2 \left( \frac{n - 1}{2} \theta \right)}{\displaystyle n + \frac{\sin \left( n \theta \right)}{\sin \left(\theta \right)}}.
\]
Now, it follows from  \eqref{eqn:v12} that  $v_{12}^2$ is an increasing function, so $v_{12}^2(w) - v_{13}^2(w) > 0$ for all $w^* \le w \le w^* + \delta$.  

Finally, suppose for some $w^* > 1$ that for all $1 \le w < w^*$, we have $\lambda_2(w) - \lambda_3(w) \ge \lambda_2(1) - \lambda_3(1)$ and $v_{12}^2(w) - v_{13}^2(w) > 0$.  Then since $\lambda_2(w)$ and $\lambda_3(w)$ are continuous functions of $w$, we obtain that $\lambda_2(w^*) - \lambda_3(w^*) \ge \lambda_2(1) - \lambda_3(1)$.  By the previous argument, $v_{12}^2(w^*) - v_{13}^2(w) > 0$.

Therefore, $\lambda_2(w) - \lambda_3(w) \ge \lambda_2(1) - \lambda_3(1)$ and $v_{12}^2(w) - v_{13}^2(w) > 0$ for all $w \ge 1$.
\end{proof}

\begin{lemma} \label{lem:evec-recur}
Let $A(w)$ be the adjacency matrix of the path on $n$ vertices with loops of weight $w$ on the end vertices.  Let $\lambda$ be an eigenvalue of $A(w)$ and consider the recurrence relation given by $a_1 = \lambda - w$, $a_k = \lambda - 1 / a_{k - 1}$.  If $v_1 = 1$ and $v_k = a_{k - 1} v_{k - 1}$, then $v = \begin{bmatrix} v_1 & v_2 & \cdots & v_n \end{bmatrix}^T$ is an eigenvector of $A(w)$ with eigenvalue $\lambda$.
\end{lemma}

\begin{proof}
If $v$ is an eigenvector of $A(w)$ with eigenvalue $\lambda$, then $v$ must satisfy $(\lambda I - A(w)) v = 0$.  We prove that applying Gaussian elimination gives the system $a_{k - 1} v_{k - 1} - v_k = 0$, $2 \le k \le n$, from which the result immediately follows.  It is clear that $k = 2$ corresponds to the first row.  Now, suppose row $j$, $1 \le j \le n - 2$ gives the equation $a_j v_j - v_{j + 1} = 0$.  Then row $j + 1$ gives the equation $-v_j + \lambda v_{j + 1} - v_{j + 2}$.  Adding $1/a_j$ times row $j$ to row $j + 1$ gives $(\lambda - 1 / a_j) v_{j + 1} - v_{j + 2} = a_{j + 1} v_{j + 1} - v_{j + 2}$ as desired.  Finally, since $\lambda$ is an eigenvalue, $(\lambda I - A(w))$ has rank less than $n$, so the last row must reduce to 0.
\end{proof}

\begin{lemma} \label{lem:gap12}
$v_{11}^2(w) - v_{12}^2(w) < 0$.
\end{lemma}

\begin{proof}
Let $y$ be the eigenvector of $\lambda_1(w)$ with $y_1 = 1$ and let $z$ be the eigenvector of $\lambda_2(w)$ with $z_1 = 1$.  Consider the recurrence relations given by $a_1 = \lambda_1 - w$, $a_k = \lambda_1 - 1 / a_{k - 1}$ and $b_1 = \lambda_2 - w$, $b_k = \lambda_2 - 1 / b_{k - 1}$.  Then by Lemma~\ref{lem:evec-recur}, we have that $y_k = a_{k - 1} y_{k - 1}$ and $z_k = b_{k - 1} z_{k - 1}$.

We will show that for all $1 \le k < n / 2$, $a_k > b_k$ and $y_{k + 1} > z_{k + 1}$.  We note that $\begin{bmatrix} y_1 & y_2 & \cdots & y_{\lceil \frac{n}{2} \rceil} \end{bmatrix}^T$ is the Perron vector for $B_1(w)$ or $C_1(w)$,  $\begin{bmatrix} z_1 & z_2 & \cdots & z_{\lfloor \frac{n}{2} \rfloor} \end{bmatrix}^T$ is the Perron vector for $B_2(w)$ or $C_2(w)$, and for odd $n$, $z_{\frac{n + 1}{2}} = 0$, therefore $a_k, b_k, y_{k + 1}, z_{k + 1} \ge 0$ for $1 \le k < n / 2$.  For $k = 1$, we obtain that $y_2 = a_1 = \lambda_1 - w > \lambda_2 - w = b_1 = z_2$.  Now suppose for some $j \ge 1$, we have $a_j > b_j \ge 0$ and $y_{j + 1} > z_{j + 1} \ge 0$.  Then we have $a_{j + 1} = \lambda_1 - 1 / a_j > \lambda_2 - 1 / b_j = b_{j + 1}$ and $y_{j + 2} = a_{j + 1} y_{j + 1} > b_{j + 1} z_{j + 1} = z_{j + 2}$.  It follows that $v_{11}^2(w) < v_{12}^2(w)$, as desired.
\end{proof}

Here is the main result of  this section. 

\begin{theorem}
Let $A(w)$ be the adjacency matrix of the path on $n \ge 3$ vertices with loops of weight $w \ge 1$ on the end vertices.  The set of weights $w$ such that $A(w)$ does not have pretty good state transfer is a dense subset of $[1,\infty)$.
\end{theorem}

\begin{proof}
Let $w^* \ge 1$.  We claim there exists a sequence $\{w_j\}$ such that $\lim_{j \to \infty} w_j = w^*$ and $A(w_j)$ does not have pretty good state transfer.  Let $\lambda_k(w)$, $1 \le k \le n$ denote the $k$-th largest eigenvalue of $A(w)$.  Consider the function
\[
h(w, x) = x \lambda_1(w) - (x + 1) \lambda_2(w) + \lambda_3(w).
\]
It has a level curve $h(w, x) = 0$ given by
\[
x = \frac{\lambda_2(w) - \lambda_3(w)}{\lambda_1(w) - \lambda_2(w)}.
\]
In particular, we note that $\lambda_1(w) > \lambda_2(w) > \lambda_3(w)$, so $x(w) > 0$, and $x$ is a continuous function of $w$ as each $\lambda_k(w)$ is a continuous function of $w$.  By Lemmas~\ref{lem:gap23} and \ref{lem:gap12}, we have that $x$ is increasing for $w \ge 1$.

Let $x^* = x(w^*)$ and let $\{x_j = \frac{p_j}{q_j}\}$ be a positive decreasing sequence such that $p_j$ is even, $q_j$ is odd, and $\lim_{j \to \infty} x_j = x^*$.  We now construct a sequence $\{w_j\}$.  We have
\begin{align*}
    h(w^*, x_1) &= h(w^*, x^*) + (x_1 - x^*) (\lambda_1(w^*) - \lambda_2(w^*)) > 0, \\
    \lim_{w \to \infty} h(w, x_1) &= -\infty,
\end{align*}
so by the Intermediate Value Theorem, and the monotonicity of $x$, there exists a unique $w_1 > w^*$ such that $h(w, x_1) = 0$.

Now, suppose for some $m \ge 1$, we have $w_{m} > w^*$ such that $h(w, x_m) = 0$.  We have
\begin{align*}
    h(w^*, x_{m + 1}) &= h(w^*, x^*) + (x_{m + 1} - x^*) (\lambda_1(w^*) - \lambda_2(w^*)) > 0, \\
    h(w_m, x_{m + 1}) &= h(w_m, x_m) + (x_{m + 1} - x_m) (\lambda_1(w_m) - \lambda_2(w_m)) < 0,
\end{align*}
so by the Intermediate Value Theorem, and the monotonicity of $x$, there exists a unique $w_{m + 1}$ such that $w_m > w_{m + 1} > w^*$ and $h(w_{m + 1}, x_{m + 1}) = 0$.  

It follows that
\[
p_{m + 1} \lambda_1(w_{m + 1}) - (p_{m + 1} + q_{m + 1}) \lambda_2(w_{m + 1}) + q_{m + 1} \lambda_3(w_{m + 1}) = 0,
\]
the sum of the coefficients is zero, and the middle coefficient is odd, so by Lemma~\ref{lem:pgst}, $A(w_{m + 1})$ does not have pretty good state transfer.

Therefore, by the Principle of Mathematical Induction, we obtain a decreasing sequence $\{w_j\}$ that converges to $w^*$ and such that $A(w_j)$ does not have pretty good state transfer.
\end{proof}

\section{Estimates for $\lambda_1$ and $\lambda_2$ and their Eigenvectors}\label{sec:l1l2}

One of our main objectives is to provide concrete estimates of the fidelity of state transfer between the end vertices of $P_n^w.$ 
In order to do that, we required detailed information on the eigenvalues and eigenvectors of the adjacency matrix $A(w)$. This section focuses on those details for the two largest eigenvalues.

\begin{theorem}\label{thm:bds} Suppose that $n\in \naturals$ with $n\ge 3$ and that $w>1.$ Denote the largest and second--largest eigenvalues of $A(w)$ by $\lambda_1$ and $\lambda_2,$ respectively. \\
a) If $n$ is even, then $\lambda_1$ and $\lambda_2$ are, respectively, the largest eigenvalues of the $\frac{n}{2} \times \frac{n}{2} $ matrices $B_1(w), B_2(w)$ in \eqref{eq:b1b2}. \\
b)  If $n$ is odd, then $\lambda_1$ and $\lambda_2$ are, respectively, the largest eigenvalues of the  matrices $C_1(w), C_2(w)$ in \eqref{eq:c1c2}, where $C_1(w)$ is of order $\frac{n+1}{2}$ and $C_2(w)$ is of order $\frac{n-1}{2}$.  \\
c) $\lambda_1 \ge w+\frac{1}{w}+\frac{(w-1)^2(w+1)}{w(w^{2 \lceil \frac{n}{2}\rceil}-1)}$ and 
$w+\frac{1}{w} \ge \lambda_2 \ge w+\frac{1}{w}-\frac{(w+1)^2(w-1)}{w(w^{2\lceil \frac{n}{2}\rceil}-1)}.$
\end{theorem}

\begin{proof}
a) Let $x$ denote a positive  $\lambda_1$--eigenvector of $A(w)$, and note that $x_j=x_{n+1-j},$ $j=1, \ldots, n.$ 
Observe now that the vector $\left[\begin{array}{cccc} x_1&x_2&\ldots&x_{\frac{n}{2}}\end{array}\right]^T$ is a $\lambda_1$--eigenvector of $B_1(w)$; further since it is a positive eigenvector and $B_1(w)$ is nonnegative, it therefore corresponds to the spectral radius of that matrix.  

Similarly, let $y$ be a $\lambda_2$--eigenvector of $A(w)$ and without loss of generality assume that $y_1>0.$ We have $y_j>0, j=1, \ldots, \frac{n}{2}$ and $y_j=-y_{n+1-j},$ $j=1, \ldots, n.$ Evidently $\left[\begin{array}{cccc} y_1&y_2&\ldots&y_{\frac{n}{2}}\end{array}\right]^T$ is a positive $\lambda_2$--eigenvector of the symmetric and essentially nonnegative matrix $B_2(w),$ and so it corresponds to the largest eigenvalue of that matrix. \\
 b) The proof is similar to that of a), with the following modifications: i) $x_{\frac{n-1}{2}} = x_{\frac{n+1}{2}}$; ii)~$y_{\frac{n+1}{2}} = 0$; iii)~$\left[\begin{array}{ccccc} x_1&x_2&\ldots&x_{\frac{n-1}{2}}&\sqrt{2}x_{\frac{n-1}{2}}\end{array}\right]^T$ is a positive $\lambda_1$--eigenvector of $C_1(w)$; and iv)~$\left[\begin{array}{cccc} y_1&y_2&\ldots&y_{\frac{n-1}{2}}\end{array}\right]^T$ is a positive $\lambda_2$--eigenvector of  $C_2(w).$ We leave the remaining details to the reader. \\
c) Consider the $k \times k$ matrix $M_k(w)$ given by 
$$M_k(w)=\left[\begin{array}{cccccc} 
w&1&0&0&\ldots&0\\
1&0&1&0&\ldots &0\\
0&1&0&1&\ldots&0\\
\vdots&\vdots & \ddots &\ddots&\ddots &\vdots \\
0&0&\ldots&1&0&1\\
0&0&\ldots&0&1&\frac{1}{w}
\end{array}\right].$$ It is straightforward to verify that the Perron value of $M_k(w)$ is $w+\frac{1}{w},$ and the vector $v$ with $v_j=\frac{1}{w^{j-1}},$ $j=1,\ldots, k$ is a corresponding Perron vector. 

Suppose that $n$ is even and consider the case $k=\frac{n}{2}.$ We have $B_1(w)= M_k(w)+e_ke_k^T,$ from which it follows that $\lambda_1 \ge w+\frac{1}{w} + \frac{v_k^2}{||v||_2^2}.$ The inequality $\lambda_1 \ge w+\frac{1}{w}+\frac{(w-1)^2(w+1)}{w(w^{n}-1)}$ now follows. 

Next, we suppose that $n$ is odd and consider the case $k=\frac{n+1}{2}$. Let $\tilde{B}_1$ be the matrix of order $k$ having the same structure as $B_1(w).$ Considering the vector $\tilde{v}$ with $\tilde v_j=v_j, j=1, \ldots, k-1, \tilde v_k=v_k/\sqrt{2},$ we find that $C_1(w) \tilde v \ge \rho(\tilde B_1) \tilde v,$ from which it follows that $\lambda_1= \rho(C_1(w)) \ge \rho(\tilde B_1).$ The desired lower bound on $\lambda_1$ now follows. 

In the case that $n$ is even we have $B_2(w)=M_{\frac{n}{2}}(w) - \frac{w+1}{w}e_{\frac{n}{2}}e_{\frac{n}{2}}^T,$
while if $n$ is odd, we obtain that 
$C_2(w)=M_{\frac{n-1}{2}}(w)-e_{\frac{n-1}{2}}e_{\frac{n-1}{2}}^T.$ In either case, using the technique used to establish a), we deduce that  $w+\frac{1}{w} \ge \lambda_2 \ge w+\frac{1}{w}-\frac{(w+1)^2(w-1)}{w(w^{2\lceil \frac{n}{2}\rceil}-1)}.$ 
\end{proof}

Next, we develop upper and lower bounds on $\lambda_1-\lambda_2$. It will be convenient to analyse the cases that $n$ is even and $n$ is odd separately. 

First suppose that $n$ is even. From Theorem \ref{thm:bds}, $\lambda_1$ and $\lambda_2$ are the largest eigenvalues of $B_1(w), B_2(w),$ respectively.  Note that $B_2(w) = B_1(w) - 2 e_{n / 2} e_{n / 2}^T$.  Let $v$ be a Perron vector of $B_1(w)$, normalized so that $||v||_2 = 1$.  Then $v^T B_2(w) v = \lambda_1 - 2 v_{n / 2}^2$, and since $\lambda_2$ is the spectral radius of $B_2$, we see that $\lambda_2 \ge \lambda_1 - 2 v_{n / 2}^2$.  Hence $\lambda_1 - \lambda_2 \le 2v_{n / 2}^2$.

Let $d_0 = 1$, $d_1 = \lambda_1 - 1$, and for $2 \le k \le n/2 - 1$,
\[
d_k = \det \left( \begin{bmatrix} \lambda_1 & -1 \\ -1 & \ddots & \ddots \\ & \ddots & \ddots & \ddots \\ & & \ddots & \lambda_1 & -1 \\ & & & -1 & \lambda_1 - 1 \end{bmatrix}_{k \times k} \right).
\]
Evidently, $d_j = \lambda_1 d_{j - 1} - d_{j - 2}$ for $2 \le j \le n/2 - 1$.  

Next we consider the case that $n$ is odd. 
Then $\lambda_1, \lambda_2$ are, respectively, the spectral radii of the matrices $C_1(w)$ (of order $\frac{n+1}{2}$) and $ C_2(w)$ (of order $\frac{n-1}{2}$).   Let 
\[
C_0(w) = C_1(w) - \sqrt{2} (e_{(n - 1) / 2} e_{(n + 1) / 2}^T + e_{(n + 1) / 2} e_{(n - 1) / 2}^T),
\]
and note that we can consider $C_0(w)$ as a block diagonal matrix with blocks $C_2(w)$ and $0$; hence, $\lambda_2$ is the spectral radius of $C_0(w)$. Then $v^T C_0(w) v = \lambda_1 - 2\sqrt{2} v_{(n - 1) / 2} v_{(n + 1) / 2}$, and since $\lambda_2$ is the spectral radius of $C_0(w)$, we see that $\lambda_2 \ge \lambda_1 - 2\sqrt{2} v_{(n - 1) / 2} v_{(n + 1) / 2}$.  Hence $\lambda_1 - \lambda_2 \le 2\sqrt{2} v_{(n - 1) / 2} v_{(n + 1) / 2}$.

Let $d_0 = 1$, $d_1 = \frac{\lambda_1}{\sqrt{2}}$, and for $2 \le k \le n/2 - 1$,
\[
d_k = \det \left( \begin{bmatrix} \lambda_1 & -1 \\ -1 & \ddots & \ddots \\ & \ddots & \ddots & -1 \\ & & -1 & \lambda_1 & - \sqrt{2} \\ & & & - 1 & \frac{\lambda_1}{\sqrt{2}} \end{bmatrix}_{k \times k} \right)
\]
Evidently, $d_2 = \frac{\lambda_1^2}{\sqrt{2}} - \sqrt{2}$ and $d_j = \lambda_1 d_{j - 1} - d_{j - 2}$ for $3 \le j \le (n - 1)/2$.  

We establish the following expression for $d_j$.

\begin{lemma}\label{lem:dj} Suppose that $n \in \naturals$ with $n \ge 2$. We have  
$$d_j =
\begin{cases} 
\frac{\sqrt{\lambda_1^2 - 4} + \lambda_1 - 2}
{2 \sqrt{\lambda_1^2 - 4}} (\frac{1}{2} (\lambda_1 + \sqrt{\lambda_1^2 - 4}))^j + \frac{\sqrt{\lambda_1^2 - 4} - \lambda_1 + 2}{2 \sqrt{\lambda_1^2 - 4}} (\frac{1}{2} (\lambda_1 - \sqrt{\lambda_1^2 - 4}))^j, & 0 \le j \le n/2 - 1, \text{ $n$ even};\\
\frac{1}{\sqrt{2}} (\frac{1}{2} (\lambda_1 + \sqrt{\lambda_1^2 - 4}))^j + \frac{1}{\sqrt{2}} (\frac{1}{2} (\lambda_1 - \sqrt{\lambda_1^2 - 4}))^j, & 1 \le j \le (n - 1)/2, \text{ $n$ odd}. 
\end{cases}$$
\end{lemma}

\begin{proof} We will only present the proof for the case that $n$ is even, as the argument when $n$ is odd is analogous. 

Taking $n$ to be even, we proceed by induction on $j$.  We have
\begin{align*}
d_0 &= \frac{\sqrt{\lambda_1^2 - 4} + \lambda_1 - 2}{2 \sqrt{\lambda_1^2 - 4}}  + \frac{\sqrt{\lambda_1^2 - 4} - \lambda_1 + 2}{2 \sqrt{\lambda_1^2 - 4}} = 1, \\
d_1 &= \frac{\sqrt{\lambda_1^2 - 4} + \lambda_1 - 2}{2 \sqrt{\lambda_1^2 - 4}} \left(\frac{1}{2} \left(\lambda_1 + \sqrt{\lambda_1^2 - 4} \right)\right) + \frac{\sqrt{\lambda_1^2 - 4} - \lambda_1 + 2}{2 \sqrt{\lambda_1^2 - 4}} \left(\frac{1}{2} \left(\lambda_1 - \sqrt{\lambda_1^2 - 4} \right)\right) = \lambda_1 - 1.
\end{align*}
Now suppose that the claim holds for all $0 \le j \le j_0 < n/2 - 1$.  Then
\begin{align*}
d_{j_0 + 1} &= \lambda_1 d_{j_0} - d_{j_0 - 1} \\
&= \lambda_1 \left( \frac{\sqrt{\lambda_1^2 - 4} + \lambda_1 - 2}{2 \sqrt{\lambda_1^2 - 4}} \left(\frac{1}{2} \left(\lambda_1 + \sqrt{\lambda_1^2 - 4} \right)\right)^{j_0} + \frac{\sqrt{\lambda_1^2 - 4} - \lambda_1 + 2}{2 \sqrt{\lambda_1^2 - 4}} \left(\frac{1}{2} \left(\lambda_1 - \sqrt{\lambda_1^2 - 4} \right)\right) ^{j_0} \right) \\
&\quad - \frac{\sqrt{\lambda_1^2 - 4} + \lambda_1 - 2}{2 \sqrt{\lambda_1^2 - 4}} \left(\frac{1}{2} \left(\lambda_1 + \sqrt{\lambda_1^2 - 4} \right)\right)^{j_0 - 1} - \frac{\sqrt{\lambda_1^2 - 4} - \lambda_1 + 2}{2 \sqrt{\lambda_1^2 - 4}} \left(\frac{1}{2} \left(\lambda_1 - \sqrt{\lambda_1^2 - 4} \right)\right) ^{j_0 - 1} \\
&= \frac{\sqrt{\lambda_1^2 - 4} + \lambda_1 - 2}{2 \sqrt{\lambda_1^2 - 4}}  \left(\frac{1}{2} \left(\lambda_1 + \sqrt{\lambda_1^2 - 4} \right)\right)^{j_0 - 1} \left(\frac{\lambda_1}{2} \left(\lambda_1 + \sqrt{\lambda_1^2 - 4} \right) - 1 \right) \\
&\quad + \frac{\sqrt{\lambda_1^2 - 4} - \lambda_1 + 2}{2 \sqrt{\lambda_1^2 - 4}}  \left(\frac{1}{2} \left(\lambda_1 - \sqrt{\lambda_1^2 - 4} \right)\right)^{j_0 - 1} \left(\frac{\lambda_1}{2} \left(\lambda_1 - \sqrt{\lambda_1^2 - 4} \right) - 1 \right) \\
&= \frac{\sqrt{\lambda_1^2 - 4} + \lambda_1 - 2}{2 \sqrt{\lambda_1^2 - 4}}  \left(\frac{1}{2} \left(\lambda_1 + \sqrt{\lambda_1^2 - 4} \right)\right)^{j_0 + 1} + \frac{\sqrt{\lambda_1^2 - 4} - \lambda_1 + 2}{2 \sqrt{\lambda_1^2 - 4}}  \left(\frac{1}{2} \left(\lambda_1 - \sqrt{\lambda_1^2 - 4} \right)\right)^{j_0 + 1},
\end{align*}
completing the induction step.
\end{proof}

Lemma \ref{lem:dj} will assist in proving the following. 

\begin{theorem} \label{thm:even} 
a) Suppose that $n$ is even. Then 
$$\frac{(w-1)^2(w+1)}{w(w^n-1)} \le \lambda_1 - \lambda_2 \le \frac{2 (w - 1) (w + 1)^3}{w^2 (w^n - 1)}.$$ In particular, $\lambda_1-\lambda_2  = \Theta(w^{2 - n}).$ \\
b) Suppose that $n$ is odd. Then $$\frac{(2 (\sqrt{2} - 1) w^2 - 1) (w^2 - 1)}{w (w^{n + 1} - 1)} \le  \lambda_1-\lambda_2 \le  \frac{4 (w - 1) (w + 1)^2}{w^2 (w^{n - 1} - 1)}.$$ 
In particular, $\lambda_1-\lambda_2 = \Theta(w^{2-n}).$
\end{theorem}
\begin{proof}
a) The lower bound on $\lambda_1-\lambda_2$ follows immediately from Theorem \ref{thm:bds} c). 

Next, we consider the upper bound  on $\lambda_1-\lambda_2.$ 
It is straightforward to determine that $$x = \begin{bmatrix} d_{n/2 - 1} & d_{n / 2 - 2} & \cdots & d_2 & d_1 & d_0 \end{bmatrix}^T$$ is a Perron vector for $B_1(w)$.  Hence, $v = \frac{x}{||x||_2}$.  We can lower bound $||x||_2^2$ by

\begin{align*}
&\quad \sum_{j = 0}^{n/2 - 1} d_j^2 \\
&= \sum_{j = 0}^{n/2 - 1} \left( \frac{\sqrt{\lambda_1^2 - 4} + \lambda_1 - 2}{2 \sqrt{\lambda_1^2 - 4}} \left(\frac{1}{2} \left(\lambda_1 + \sqrt{\lambda_1^2 - 4}\right)\right)^j + \frac{\sqrt{\lambda_1^2 - 4} - \lambda_1 + 2}{2 \sqrt{\lambda_1^2 - 4}} \left(\frac{1}{2} \left(\lambda_1 - \sqrt{\lambda_1^2 - 4}\right)\right)^j \right)^2 \\
&= \sum_{j = 0}^{n/2 - 1} \left( \frac{\lambda_1 + \sqrt{\lambda_1^2 - 4}}{2 \lambda_1 + 4} \left(\frac{1}{2} \left(\lambda_1 + \sqrt{\lambda_1^2 - 4}\right)\right)^{2j}
 + \frac{2}{\lambda_1 + 2}
 +  \frac{\lambda_1 - \sqrt{\lambda_1^2 - 4}}{2 \lambda_1 + 4} \left(\frac{1}{2} \left(\lambda_1 - \sqrt{\lambda_1^2 - 4}\right)\right)^{2j} \right) \\
&=\frac{\lambda_1 + \sqrt{\lambda_1^2 - 4}}{2 \lambda_1 + 4} \left( \frac{(\frac{1}{2} (\lambda_1 + \sqrt{\lambda_1^2 - 4}))^n - 1}{(\frac{1}{2} (\lambda_1 + \sqrt{\lambda_1^2 - 4}))^2 - 1} \right)
 + \frac{n}{\lambda_1 + 2} 
+ \frac{\lambda_1 - \sqrt{\lambda_1^2 - 4}}{2 \lambda_1 + 4} \left( \frac{(\frac{1}{2} (\lambda_1 - \sqrt{\lambda_1^2 - 4}))^n - 1}{(\frac{1}{2} (\lambda_1 - \sqrt{\lambda_1^2 - 4}))^2 - 1} \right) \\
&\ge \frac{\lambda_1 + \sqrt{\lambda_1^2 - 4}}{2 \lambda_1 + 4} \left( \frac{(\frac{1}{2} (\lambda_1 + \sqrt{\lambda_1^2 - 4}))^n - 1}{(\frac{1}{2} (\lambda_1 + \sqrt{\lambda_1^2 - 4}))^2 - 1} \right) \\
&\ge \frac{w + w^{-1} + \sqrt{(w + w^{-1})^2 - 4}}{2 (w + w^{-1}) + 4} \left( \frac{(\frac{1}{2} (w + w^{-1} + \sqrt{(w + w^{-1})^2 - 4}))^n - 1}{(\frac{1}{2} (w + w^{-1} + \sqrt{(w + w^{-1})^2 - 4}))^2 - 1} \right) \\
&= \frac{2w}{2 w + 2 w^{-1} + 4} \left( \frac{w^n - 1}{w^2 - 1} \right) \\
&= \frac{w^2 (w^n - 1)}{(w - 1) (w + 1)^3}
\end{align*}
by expanding, applying geometric series, dropping the other positive terms, and using that the remaining term is increasing on $\lambda_1 \ge 2$ together with the fact that $\lambda_1 \ge w + w^{-1}$.  Therefore
\[
\lambda_1 - \lambda_2 \le 2 v_{n / 2}^2 = 2 \frac{d_0^2}{\sum_{j = 0}^{n / 2 - 1} d_j^2} \le \frac{2 (w - 1) (w + 1)^3}{w^2 (w^n - 1)} .
\]

b) The proof of the upper bound on $\lambda_1-\lambda_2$ is similar to that for the upper bound in a), and consequently we omit it. 

Hence, we turn our attention to the lower bound on $\lambda_1-\lambda_2.$ 
Recall that 
\[
M(w) = \begin{bmatrix} w & 1 \\ 1 & 0 & 1 \\ & \ddots & \ddots & \ddots \\ & & 1 & 0 & 1 \\ & & & 1 & \frac{1}{w} \end{bmatrix}_{k \times k},
\]
 has Perron value $w + \frac{1}{w}$ and corresponding Perron vector $x = \begin{bmatrix} w^{n - 1} & w^{n - 2} & \cdots & w & 1 \end{bmatrix}^T$.  We can write $C_1(w) = M(w) - \frac{1}{w} e_k e_k^T + (\sqrt{2} - 1)(e_{k - 1} e_k^T + e_k e_{k - 1}^T)$, and it follows that $\lambda_1 \ge w + \frac{1}{w} - \frac{1}{w} \frac{x_k^2}{||x||_2^2} + 2 (\sqrt{2} - 1) \frac{x_{k - 1} x_{k}}{||x||_2^2}$.  Since
\[
||x||_2^2 = \sum_{j = 0}^{k - 1} w^{2j} = \frac{w^{2k} - 1}{w^2 - 1}
\]
we obtain
\begin{align*}
\lambda_1 \ge w + \frac{1}{w} - \frac{1}{w} \frac{x_n^2}{||x||_2^2} + 2 (\sqrt{2} - 1) \frac{x_{k - 1} x_{k}}{||x||_2^2} &= w + \frac{1}{w} - \frac{1}{w} \frac{w^2 - 1}{w^{2k} - 1} + 2 (\sqrt{2} - 1) w \frac{w^2 - 1}{w^{2k} - 1} \\ &= w + \frac{1}{w} + \frac{(2 (\sqrt{2} - 1) w^2 - 1) (w^2 - 1)}{w (w^{2k} - 1)}.
\end{align*}

We now deduce that for $n$ odd, 
\[
\lambda_1 - \lambda_2 \ge \frac{(2 (\sqrt{2} - 1) w^2 - 1) (w^2 - 1)}{w (w^{n + 1} - 1)}
\]
and hence $\lambda_1 - \lambda_2 = \Theta(w^{2 - n})$.
\end{proof}

Next, we consider the eigenvectors associated with $\lambda_1 $ and $\lambda_2$. 

\begin{theorem}\label{thm:v11_v12}
Consider the matrix $A(w)$ of order $n$, and let $V$ be an orthogonal matrix that diagonalises $A(w)$, i.e. $V^TA(w)V ={\rm{diag}}(\lambda_1, \ldots, \lambda_n)$. We have 
\begin{align*} 
v_{11}^2 &= \left(2+ \frac{8}{n-1}\sum_{\ell=1}^{\lceil\frac{n-1}{2}\rceil } \frac{\sin^2\left(\frac{(2\ell-1)\pi}{n-1}\right)}{\left(\lambda_1 - 2\cos\left(\frac{(2\ell-1)\pi}{n-1}\right)\right)^2} \right)^{-1}, \\
v_{12}^2 &= \left(2+ \frac{8}{n-1}\sum_{\ell=1}^{\lceil\frac{n-2}{2}\rceil } \frac{\sin^2\left(\frac{2\ell\pi}{n-1}\right)}{\left(\lambda_2 - 2\cos\left(\frac{2\ell\pi}{n-1}\right)\right)^2} \right)^{-1}. 
\end{align*} 
\end{theorem}
\begin{proof}
Let $Q$ be the orthogonal matrix of order $n-2$  given by $$Q=\left[ \begin{array}{c} \sqrt{\frac{2}{n-1}}\sin \left( \frac{\ell j \pi}{n-1}\right)\end{array} \right]_{\ell, j=1,\ldots, n-2},$$ and let $D$ be the diagonal matrix whose diagonal entries are $2 \cos\left( \frac{ j \pi}{n-1}\right), j=1, \ldots, n-2.$ It is well known that $Q$ diagonalises the adjacency matrix of the path on $n-2$ vertices, and that the corresponding eigenvalues are the diagonal entries of $D$. Consequently, $A(w)$ can be written as 
$$A(w)= \left[ \begin{array}{c|c|c} w & e_1^T&0\\ \hline 
e_1 &QDQ^T&e_{n-2}\\ \hline 0&e_{n-2}^T &w\end{array}\right].$$ 

Consider the $\lambda_1$--eigenvector of $A(w)$ whose first entry is $1$. Necessarily that eigenvector has the form $\left[ \begin{array}{c}1 \\ \hline u \\ \hline 1 \end{array}\right] .$ From the eigenequation, we find that $e_1+e_{n-2} + QDQ^T=\lambda_1 u,$ from which we deduce that $u=Q(\lambda_1-D)^{-1}Q^T(e_1+e_{n-2}).$
 Therefore $u^Tu=(e_1+e_{n-2})^TQ(\lambda_1-D)^{-2}Q^T(e_1+e_{n-2}),$ and since $e_{n-2}^TQe_\ell = (-1)^{\ell-1}e_{1}^TQe_\ell,$ $\ell =1, \ldots, n-2,$ then
$(e_1+e_{n-2})^TQ e_\ell = (1+(-1)^{\ell-1})e_1^TQe_\ell,$  $\ell =1, \ldots, n-2.$
We thus find that 
$$u^Tu = \sum_{j=1}^{n-2} \left(   (1+(-1)^{j-1})\sqrt{\frac{2}{n-1}} 
\frac{\sin\left(\frac{\pi j}{n-1}\right)}{\left(\lambda_1 - 2\cos\left(\frac{\pi j}{n-1}\right)\right)}
\right)^2 
=\frac{8}{n-1}\sum_{\ell=1}^{\lceil\frac{n-1}{2}\rceil } \frac{\sin^2\left(\frac{(2\ell-1)\pi}{n-1}\right)}{\left(\lambda_1 - 2\cos\left(\frac{(2\ell-1)\pi}{n-1}\right)\right)^2}.
$$
The desired expression for $v_{11}^2$ now follows from the fact that $v_{11}^2=\frac{1}{2+u^Tu}.$

The derivation of the expression for $v_{12}^2$ proceeds along similar lines, starting from the observation that the $\lambda_2$--eigenvector for $A(w)$ whose first entry is $1$ has the form 
$\left[ \begin{array}{c}1 \\ \hline \tilde u \\ \hline -1 \end{array}\right] ;$ arguing as above, we obtain that 
$\tilde u^T \tilde u=(e_1-e_{n-2})^TQ(\lambda_2-D)^{-2}Q^T(e_1-e_{n-2}).$ The desired expression for $v_{12}^2 =\frac{1}{2+\tilde u^T \tilde u}$ now follows readily. 
\end{proof}

\begin{lemma}\label{lem:limit} Suppose that $x_n$ is a sequence of positive numbers such that $x_n \rightarrow x>2$ as $n \rightarrow \infty.$ Then:\\  
\begin{align*}
\text{a) } &\lim_{n \rightarrow \infty}\left( 2+ \frac{8}{n-1}\sum_{\ell=1}^{\lceil\frac{n-1}{2}\rceil } \frac{\sin^2\left(\frac{(2\ell-1)\pi}{n-1}\right)}{\left(x_n - 2\cos\left(\frac{(2\ell-1)\pi}{n-1}\right)\right)^2}\right) = \frac{x+\sqrt{x^2-4}}{\sqrt{x^2-4}}; \text{ and} \\
\text{b) } &\lim_{n \rightarrow \infty}\left( 2+ \frac{8}{n-1}\sum_{\ell=1}^{\lceil\frac{n-2}{2}\rceil } \frac{\sin^2\left(\frac{2\ell\pi}{n-1}\right)}{\left(x_n - 2\cos\left(\frac{2\ell\pi}{n-1}\right)\right)^2}\right)= \frac{x+\sqrt{x^2-4}}{\sqrt{x^2-4}}.
\end{align*}
\end{lemma}
\begin{proof}
We give the details for a) and note that the proof for b) is similar. 

First we note that 
\begin{align*} 
&\quad \frac{8}{n-1}\sum_{\ell=1}^{\lceil\frac{n-1}{2}\rceil } \frac{\sin^2\left(\frac{(2\ell-1)\pi}{n-1}\right)}{\left(x_n - 2\cos\left(\frac{(2\ell-1)\pi}{n-1}\right)\right)^2} - \frac{8}{n-1}\sum_{\ell=1}^{\lceil\frac{n-1}{2}\rceil } -  \frac{\sin^2\left(\frac{(2\ell-1)\pi}{n-1}\right)}{\left(x - 2\cos\left(\frac{(2\ell-1)\pi}{n-1}\right)\right)^2} \\ &= 
\frac{8(x-x_n)}{n-1}\sum_{\ell=1}^{\lceil\frac{n-1}{2}\rceil } \sin^2\left(\frac{(2\ell-1)\pi}{n-1}\right)
\left(\frac{x+x_n-4\cos\left(\frac{(2\ell-1)\pi}{n-1}\right)}{\left(x_n - 2\cos\left(\frac{(2\ell-1)\pi}{n-1}\right)\right)^2\left(x - 2\cos\left(\frac{(2\ell-1)\pi}{n-1}\right)\right)^2}
\right).
\end{align*} 
As $\bigg|\frac{x+x_n-4\cos\left(\frac{(2\ell-1)\pi}{n-1}\right)}{\left(x_n - 2\cos\left(\frac{(2\ell-1)\pi}{n-1}\right)\right)^2\left(x - 2\cos\left(\frac{(2\ell-1)\pi}{n-1}\right)\right)^2} \bigg| \le \frac{x+x_n+4}{(x_n-2)^2(x-2)^2}$, $\ell=1, \ldots, \lceil\frac{n-1}{2}\rceil , $ and 
$\frac{8}{n-1}\sum_{\ell=1}^{\lceil\frac{n-1}{2}\rceil } \sin^2\left(\frac{(2\ell-1)\pi}{n-1}\right)$ is convergent, it suffices to show that $\lim_{n \rightarrow \infty}\left( 2+ \frac{8}{n-1}\sum_{\ell=1}^{\lceil\frac{n-1}{2}\rceil } \frac{\sin^2\left(\frac{(2\ell-1)\pi}{n-1}\right)}{\left(x - 2\cos\left(\frac{(2\ell-1)\pi}{n-1}\right)\right)^2}\right) = \frac{x+\sqrt{x^2-4}}{\sqrt{x^2-4}}.$ 

Next, observe that $\frac{2\pi }{n-1}\sum_{\ell=1}^{\lceil\frac{n-1}{2}\rceil } \frac{\sin^2\left(\frac{(2\ell-1)\pi}{n-1}\right)}{\left(x - 2\cos\left(\frac{(2\ell-1)\pi}{n-1}\right)\right)^2}$ is a Riemann sum for $\int_0^\pi \frac{\sin^2 \theta}{(x-2\cos \theta)^2} d\theta.$ We now evaluate that integral. We have $\int_0^\pi \frac{\sin^2 \theta}{(x-2\cos \theta)^2} d\theta = \frac{1}{x^2} \int_0^\pi\frac{\sin^2 \theta}{(1-\frac{2}{x}\cos \theta)^2} = 
\frac{1}{x^2} \sum_{j=0}^\infty (j+1)\left(\frac{2}{x}\right)^j \int_0^\pi \sin^2 \theta \cos^j \theta d\theta.$ 
When $j$ is odd, $\int_0^\pi \sin^2 \theta \cos^j \theta d\theta=0,$ while when $j$ is even, say $j=2\ell, \int_0^\pi \sin^2 \theta \cos^j \theta d\theta = \frac{(2 \ell+2)! \pi}{2^{2\ell+2}(\ell+1)!(\ell+1)!}
.$
Now for $z$ with $|z|<\frac{1}{4}, \sum_{k=0}^\infty \binom{2k}{k} z^k =\frac{1}{\sqrt{1-4z}},$ 
so
$\int_0^\pi \frac{\sin^2 \theta}{(x-2\cos \theta)^2} d\theta = 
\frac{\pi}{4} \sum_{\ell=0}^\infty \left( \frac{1}{x^2}\right)^{\ell+1} \binom{2 \ell+2}{\ell+1} = \frac{\pi}{4} \left( \frac{x}{\sqrt{x^2-4}} -1\right).$ The expression for the desired limit now follows via an uninteresting computation. 
\end{proof}

We have the following asymptotic result. 

\begin{theorem}\label{thm:eveclim} For the $n \times n$ matrix 
 $A(w)$,  let $V(n)$ be an orthogonal matrix that diagonalises it, i.e. $V(n)^TA(w)V(n)  =\rm{diag}(\lambda_1, \ldots, \lambda_n)$.
Then 
$$ \lim_{n\rightarrow \infty} v_{11}(n)^2 = \frac{w^2-1}{2w^2} =  \lim_{n\rightarrow \infty} v_{12}(n)^2 . 
$$
\end{theorem} 

\begin{proof}
From Theorems \ref{thm:bds}, \ref{thm:even} and \ref{thm:v11_v12}, we find that as $n\rightarrow \infty, \lambda_1, \lambda_2 \rightarrow w+\frac{1}{w}.$ The desired limits for $v_{11}(n)^2$ and $v_{12}(n)^2$ now follow from Theorem \ref{thm:v11_v12} and Lemma \ref{lem:limit}. 
\end{proof}

\begin{remark}\label{rmk:x*} 
Suppose that $\lambda>2$, and consider the function $f(x)=\left( \frac{\sin x}{\lambda - 2\cos x}\right)^2$ on the interval $x \in [0,\pi].$ It is a straightforward exercise to determine that for the value $x^*=\cos^{-1}\left(\frac{2}{\lambda}\right),$ a) $f$ is increasing on $[0,x^*],$ b) 
$f$ is decreasing on $[x^*, \pi]$, and c) $ \max \{ f(x) \mid x \in [0,\pi]\}=f(x^*)=\frac{1}{\lambda^2-4}.$
\end{remark}

This next technical result will enable us to estimate $v_{11}^2$ and $v_{12}^2.$

\begin{lemma}\label{lem:riemann} Suppose that $\lambda>2, n\in \naturals$ and $n \ge 3.$ Let $f$ be as in Remark \ref{rmk:x*}.  Then if $n$ is odd we have 
\begin{equation*}
\Bigg|\frac{2 \pi}{n-1} \sum_{\ell=1}^{\lfloor \frac{n-1}{2}\rfloor} f\left(\frac{(2\ell-1)\pi}{n-1}\right) -\frac{\pi}{4}\left( \frac{\lambda}{\sqrt{\lambda^2-4}}-1 \right)  \Bigg| \le \frac{4 \pi}{(n-1)(\lambda^2-4)},
\end{equation*}
while if $n$ is even we have 
\begin{equation*}
\Bigg|\frac{2 \pi}{n-1} \sum_{\ell=1}^{\lfloor \frac{n-2}{2}\rfloor} f\left(\frac{2\ell\pi}{n-1}\right) -\frac{\pi}{4}\left( \frac{\lambda}{\sqrt{\lambda^2-4}} -1\right) \Bigg|  \le \frac{4 \pi}{(n-1)(\lambda^2-4)}.  
\end{equation*}
\begin{proof}
We give the proof only for the case that $n$ is odd, and note that the proof when $n$ is even proceeds analogously. 

First, note that $\int_0^\pi f(x) dx=\frac{4 \pi}{(n-1)(\lambda^2-4)}.$ Next, recall that for a continuous monotonic function $g(x)$ on an interval $[a,b],$ the partition $x_j=a+\frac{(b-a)j}{k}, j=0, \ldots, k,$ and corresponding left and right Riemann sums $R,L$ respectively, we have $|L-\int_a^b g(x) dx|, |R-\int_a^b g(x) dx| \le |g(b)-g(a)|\frac{b-a}{k}.$ Our plan is to subdivide $[0,\pi]$ into subintervals on which (with one exception) $f$ is monotonic. 

Set $x^*=\cos^{-1}\left(\frac{2}{\lambda}\right)$ 
and define $\ell_0\in \naturals$ via 
$\frac{(2\ell_0-1)\pi}{n-1} \le x^* <\frac{(2\ell_0+1)\pi}{n-1}.$ Suppose that  $1\le \ell_0 \le \lfloor \frac{n-1}{2}\rfloor -1.$ 
We have 
\begin{align*}
&\quad \frac{2 \pi}{n-1} \sum_{\ell=1}^{\lfloor \frac{n-1}{2}\rfloor} f\left(\frac{(2\ell-1)\pi}{n-1}\right) - \int_0^\pi f(x) dx \\
&= \left( \frac{2 \pi}{n-1} \sum_{\ell=1}^{\ell_0-1} f\left(\frac{(2\ell-1)\pi}{n-1}\right) - \int_0^{\frac{(2\ell_0-1)\pi}{n-1}} f(x) dx\right) 
+\left( \frac{2 \pi}{n-1} \sum_{\ell=\ell_0+1}^{\lfloor \frac{n-1}{2}\rfloor} f\left(\frac{(2\ell-1)\pi}{n-1}\right) - \int_{\frac{(2\ell_0+1)\pi}{n-1}}^\pi f(x) dx \right)\\
&\quad +\left(\frac{2 \pi}{n-1}f\left(\frac{(2\ell_0-1)\pi}{n-1}\right)- \int_{\frac{(2\ell_0-1)\pi}{n-1}}^{\frac{(2\ell_0+1)\pi}{n-1}} f(x) dx \right). 
\end{align*}
Observing that $f$ is increasing on $[0,{\frac{(2\ell_0-1)\pi}{n-1}}]$ and decreasing on 
$[{\frac{(2\ell_0+1)\pi}{n-1}}, \pi]$ we find that 
$$\Bigg|\left( \frac{2 \pi}{n-1} \sum_{\ell=1}^{\ell_0-1} f\left(\frac{(2\ell-1)\pi}{n-1}\right) - \int_0^{\frac{(2\ell_0-1)\pi}{n-1}} f(x) dx\right)  \Bigg| \le \frac{2 \pi}{n-1} f\left({\frac{(2\ell_0-1)\pi}{n-1}} \right)$$
and $$\Bigg|  \frac{2 \pi}{n-1} \sum_{\ell=\ell_0+1}^{\lfloor \frac{n-1}{2}\rfloor} f\left(\frac{(2\ell-1)\pi}{n-1}\right) - \int_{\frac{(2\ell_0+1)\pi}{n-1}}^\pi f(x) dx \Bigg| \le \frac{2 \pi}{n-1} f\left({\frac{(2\ell_0+1)\pi}{n-1}} \right).$$ 

Next we consider $\int_{\frac{(2\ell_0-1)\pi}{n-1}}^{\frac{(2\ell_0+1)\pi}{n-1}} f(x) dx$. Setting $\mu = \min\{f\left(\frac{(2\ell_0-1)\pi}{n-1}\right), f\left(\frac{(2\ell_0+1)\pi}{n-1}\right) \},$ 
we see that $\mu \le f(x) \le f(x^*)$ for $x \in [\frac{(2\ell_0-1)\pi}{n-1}, \frac{(2\ell_0+1)\pi}{n-1}].$ It now follows readily that $$\Bigg|\left(\frac{2 \pi}{n-1}f\left(\frac{(2\ell_0-1)\pi}{n-1}\right)- \int_{\frac{(2\ell_0-1)\pi}{n-1}}^{\frac{(2\ell_0+1)\pi}{n-1}} f(x) dx \right)\Bigg| \le \frac{2\pi}{n-1}(f(x^*)-\mu).$$

Assembling the observations above, we find that 
\begin{align*}
\Bigg|\frac{2 \pi}{n-1} \sum_{\ell=1}^{\lfloor \frac{n-1}{2}\rfloor} f\left(\frac{(2\ell-1)\pi}{n-1}\right) -\int_0^\pi f(x) dx \Bigg| &\le \frac{2\pi}{n-1}\left( f\left(\frac{(2\ell_0-1)\pi}{n-1}\right) + f\left(\frac{(2\ell_0+1)\pi}{n-1}\right)+f(x^*)-\mu \right) \\ &\le \frac{4\pi}{n-1}f(x^*). 
\end{align*}
A similar argument applies if either $\ell_0=0$ or $\ell_0= \lfloor \frac{n-1}{2}\rfloor.$
\end{proof}
\end{lemma}

The following is immediate from Lemma \ref{lem:riemann}. 

\begin{cor}\label{cor:vbounds} Suppose that $n-1\ge \frac{16}{\sqrt{\lambda_2^2-4}(\lambda_2 + \sqrt{\lambda_2^2-4})}.$ Then 
\begin{eqnarray*}
\left( \frac{\lambda_1+\sqrt{\lambda_1^2-4}}{\sqrt{\lambda_1^2-4}} +\frac{16}{(n-1)(\lambda_1^2-4)}\right)^{-1} \le v_{11}^2 \le \left( \frac{\lambda_1+\sqrt{\lambda_1^2-4}}{\sqrt{\lambda_1^2-4}} -\frac{16}{(n-1)(\lambda_1^2-4)}\right)^{-1}, \\
\left( \frac{\lambda_2+\sqrt{\lambda_2^2-4}}{\sqrt{\lambda_2^2-4}} +\frac{16}{(n-1)(\lambda_2^2-4)}\right)^{-1} \le v_{12}^2 \le \left( \frac{\lambda_2+\sqrt{\lambda_2^2-4}}{\sqrt{\lambda_2^2-4}} -\frac{16}{(n-1)(\lambda_2^2-4)}\right)^{-1}.
\end{eqnarray*}
\end{cor}

\begin{remark}\label{rmk:bfunction}
Consider the function $b(\lambda)=\left( \frac{\lambda+\sqrt{\lambda^2-4}}{\sqrt{\lambda^2-4}} +\frac{16}{(n-1)(\lambda^2-4)}\right)^{-1} =\frac{\sqrt{\lambda^2-4}}{\lambda + \sqrt{\lambda^2-4} + \frac{16}{(n-1)\sqrt{\lambda^2-4}}} $ for $\lambda >2.$ It is straightforward to show that this is increasing as a function of $\lambda$ for $\lambda \in (2, \infty)$.  

Since $\lambda_1 > w+\frac{1}{w},$ then $v_{11}^2 \ge b(\lambda_1) \ge b(w+\frac{1}{w})= \frac{w^2-1}{2w^2}-\frac{4}{(n-1)w^2}\left[ \frac{1}{1+\frac{8}{(n-1)(w^2-1)}}\right] >\frac{w^2-1}{2w^2}-\frac{4}{(n-1)w^2} .$ 

Setting $\tau_n=\frac{(w+1)^2(w-1)}{w(w^{2\lceil \frac{n}{2}\rceil}-1)},$ we recall that $\lambda_2 \ge w+\frac{1}{w}-\tau_n.$ Suppose that $w+\frac{1}{w}-\tau_n \ge 2.$  It now follows as above that $v_{12}^2 \ge b(\lambda_2)\ge b(w+\frac{1}{w}-\tau_n).$ An uninteresting computation reveals that 
\begin{align*}
&\quad b\left(w+\frac{1}{w}-\tau_n\right) \\
&= b\left(w+\frac{1}{w}\right) + \frac{\tau_n}{(2w+\frac{16}{(n-1)(w-\frac{1}{w})})(w+\frac{1}{w}-\tau_n +\sqrt{(w+\frac{1}{w}-\tau_n)^2-4}+\frac{16}{(n-1)\sqrt{(w+\frac{1}{w}-\tau_n)^2-4}})}  \\
&\quad \times \left[ 
w-\frac{1}{w} + \frac{(w+\frac{1}{w})(-2(w+\frac{1}{w})+\tau_n)}{w-\frac{1}{w} + \sqrt{(w+\frac{1}{w}-\tau_n)^2-4}}
+\frac{16(-2(w+\frac{1}{w})+\tau_n)}{(n-1)(w-\frac{1}{w})\sqrt{(w+\frac{1}{w}-\tau_n)^2-4}}
\right].  
\end{align*} 
\end{remark}

\section{Estimating the Remaining Eigenvalues and their Eigenvectors}\label{sec:other} 

In this section, we discuss the smaller eigenvalues and the corresponding eigenvectors for $A(w)$. We begin by recalling a standard result and a useful fact from \cite{KLY}. 

\begin{theorem}[Weyl's Inequality] \label{Weyl}
If $A$ and $B$ are Hermitian, then $\lambda_{i + j - 1}(A + B) \le \lambda_i(A) + \lambda_j(B)$, and equivalently, $\lambda_{i + j - n}(A + B) \ge \lambda_i(A) + \lambda_j(B)$.
\end{theorem}

\begin{lemma} \cite{KLY} \label{lem:invol}
Let $G$ be a (weighted) graph with an involution $\sigma$, which respects loops and edge weights.  Then the characteristic polynomial of the (weighted) adjacency matrix $A$ of $G$ factors into two factors $P_+$ and $P_-$ which are, respectively, the characteristic polynomials of $A_+ := \begin{bmatrix} A' + A_\sigma & A_\delta \\ 2 A_\delta^T & A_S \end{bmatrix}$ and $A_- := A' - A_\sigma$.

Furthermore, there is an eigenbasis for $A$ consisting of vectors that take the form $\begin{bmatrix} a & a & b \end{bmatrix}^T$ and $\begin{bmatrix} c & -c & 0 \end{bmatrix}^T$, where $\begin{bmatrix} a & b \end{bmatrix}^T$ is an eigenvector for $A_+$, and $c$ an eigenvector for $A_-$.
\end{lemma}

For the weighted path with involution $\sigma : j \mapsto n + 1 - j$, we have $A' = A(P_{\lfloor \frac{n}{2} \rfloor}) + w e_1 e_1^T$; for $n$ even, we have $A_\sigma = e_{n/2} e_{n / 2}^T$ with $A_\delta, A_S$ as empty matrices, while for $n$ odd, we have $A_\sigma = 0$, $A_\delta = e_{\lfloor \frac{n}{2} \rfloor}$ and $A_S = 0$.

\begin{lemma} \label{lem:interlace}
Consider the matrix $A(w)$ of order $n$, and let $\lambda_1 \ge \lambda_2 \ge \cdots \ge \lambda_n$ be the eigenvalues of $A(w)$.  If $v$ is an eigenvector corresponding to $\lambda_j$ for $j$ odd, then $v_\ell = v_{n + 1 - \ell}$, and if $u$ is an eigenvector corresponding to $\lambda_j$ for $j$ even, then $u_\ell = - u_{n + 1 - \ell}$.
\end{lemma}

\begin{proof}
By Lemma~\ref{lem:invol}, the eigenvalues of $A(w)$ are given by the eigenvalues of $A_+$, which correspond to eigenvectors $v$ satisfying $v_\ell = v_{n + 1 - \ell}$, and the eigenvalues of $A_-$, which correspond to eigenvectors $u$ satisfying $u_\ell = - u_{n + 1 - \ell}$.  It remains to show the eigenvalues of $A_+$ and $A_-$ interlace.  

If $x$ is an eigenvector of $A(w)$ and $x_1 = 0$, then $x = 0$, a contradiction.  If $x$ and $y$ are both eigenvectors of $A(w)$ corresponding to the same eigenvalue $\lambda$, then $x - \frac{x_1}{y_1} y$ is also an eigenvector with eigenvalue $\lambda$, whose first component is zero, a contradiction.  Hence, the eigenvalues of $A(w)$ are distinct.

For $n$ odd, $A_-$ is a principal submatrix of $A_+$, so the eigenvalues of $A_+$ interlace the eigenvalues of $A_-$ as claimed.  For $n$ even, we have $A_+ = A_- + 2 e_{n/2} e_{n/2}^T$.  Hence by Weyl's Inequality, we have $\lambda_{i}(A_+) \ge \lambda_{i}(A_-) + \lambda_{n/2}(e_{n/2} e_{n/2}^T) \ge \lambda_i(A_-)$ and $\lambda_{i + 1}(A_+) \le \lambda_i(A_-) + \lambda_2(e_{n/2} e_{n/2}^T) = \lambda_i(A_-)$.  Hence, the eigenvalues of $A_+$ and $A_-$ interlace as claimed.
\end{proof}

Next we establish expressions for the eigenvector entries. 

\begin{lemma} \label{lem:evecs}
Consider the matrix $A(w)$ of order $n$, and for some $\theta \in [0, \pi]$, suppose $2 \cos (\theta)$ is the $j$-th eigenvalue of $A(w)$, for $j=3, \ldots, n$.  Then 
\[
v_{\ell} = \cos \left( \left( \ell - \frac{n + 1}{2} \right) \theta \right), \ell=1, \ldots, n 
\]
is the corresponding eigenvector if $j$ is odd, and
\[
u_{\ell} = \sin \left( \left( \ell - \frac{n + 1}{2} \right) \theta \right), \ell=1, \ldots, n
\]
is the corresponding eigenvector if $j$ is even.  Moreover, $2 \cos (\theta)$ is the $j$-th eigenvalue of $A(w)$ if and only if
\[
w \cos \left( \left( \frac{n - 1}{2} \right) \theta \right) = \cos \left( \left( \frac{n + 1}{2} \right) \theta \right)
\]
for $j$ odd or
\[
w \sin \left( \left( \frac{n - 1}{2} \right) \theta \right) = \sin \left( \left( \frac{n + 1}{2} \right) \theta \right)
\]
for $j$ even.
\end{lemma}

\begin{proof}
Here we only give the proof for the case that $j$ is odd, as the argument when $j$ is even is analogous. 

Suppose $j$ is odd, then by Lemma~\ref{lem:interlace}, the eigenvector satisfies $v_{\ell} = v_{n + 1 - \ell}$.  If $n$ is even, then $v_{\frac{n}{2}} = v_{\frac{n}{2} + 1} \neq 0$, as otherwise $v = 0$.  Let $v$ be the eigenvector with $v_{\frac{n}{2}} = v_{\frac{n}{2} + 1} = \cos \left( \frac{\theta}{2} \right)$.  We will prove that $v_\ell = \cos \left( \left( \ell - \frac{n - 1}{2} \right) \theta \right)$.  It is clear the result holds for $\ell = n/2, n/2 + 1$.  Now suppose for some $1 \le \ell < n / 2$, the result holds for all $k$ such that $\ell < k \le n / 2$.  Then we have
\begin{align*}
    [A(w) v]_{\ell + 1} &= 2 \cos (\theta) v_{\ell + 1}, {\mbox{\rm{i.e.,}}} \\
    v_{\ell} + v_{\ell + 2} &= 2 \cos (\theta) v_{\ell + 1}, {\mbox{\rm{equivalently}}} \\
    v_{\ell} &= 2 \cos (\theta) v_{\ell + 1} - v_{\ell + 2}, {\mbox{\rm{equivalently}}} \\
    v_{\ell} &= 2 \cos (\theta) \cos \left( \left( \ell + 1 - \frac{n - 1}{2} \right) \theta \right) - \cos \left( \left( \ell + 2 - \frac{n - 1}{2} \right) \theta \right), {\mbox{\rm{equivalently}}} \\
    v_{\ell} &= \cos \left( \left( \ell - \frac{n - 1}{2} \right) \theta \right) + \cos \left( \left( \ell + 2 - \frac{n - 1}{2} \right) \theta \right) - \cos \left( \left( \ell + 2 - \frac{n - 1}{2} \right) \theta \right), {\mbox{\rm{and finally}}} \\
    v_{\ell} &= \cos \left( \left( \ell - \frac{n - 1}{2} \right) \theta \right), 
\end{align*}
as desired.  If $n$ is odd, then $v_{\frac{n + 1}{2}} \neq 0$, as otherwise $v = 0$.  Let $v$ be the eigenvector with $v_{\frac{n + 1}{2}} = 1$.  Then we have
\begin{align*}
    [A(w) v]_{\frac{n + 1}{2}} &= 2 \cos (\theta) v_{\frac{n + 1}{2}}, {\mbox{\rm{i.e.,}}} \\
    v_{\frac{n - 1}{2}} + v_{\frac{n + 3}{2}} &= 2 \cos (\theta), {\mbox{\rm{which yields}}} \\
    2 v_{\frac{n - 1}{2}} &= 2 \cos (\theta), {\mbox{\rm{and hence}}} \\
    v_{\frac{n - 1}{2}} &= \cos (\theta),
\end{align*}
and the result follows as above.  Moreover, $2 \cos (\theta)$ is the $j$-th eigenvalue of $A(w)$ if and only if
\begin{align*}
    w v_n + v_{n - 1} &= 2 \cos (\theta) v_n, {\mbox{\rm{which is equivalent to, respectively, }}} \\
    w \cos \left( \left( \frac{n - 1}{2} \right) \theta \right) + \cos \left( \left( \frac{n - 3}{2} \right) \theta \right) &= 2 \cos (\theta) \cos \left( \left( \frac{n - 1}{2} \right) \theta \right), \\
    w \cos \left( \left( \frac{n - 1}{2} \right) \theta \right) + \cos \left( \left( \frac{n - 3}{2} \right) \theta \right) &= \cos \left( \left( \frac{n - 3}{2} \right) \theta \right) + \cos \left( \left( \frac{n + 1}{2} \right) \theta \right), {\mbox{\rm{and}}} \\
    w \cos \left( \left( \frac{n - 1}{2} \right) \theta \right) &=  \cos \left( \left( \frac{n + 1}{2} \right) \theta \right), 
\end{align*}
as desired.
\end{proof}

Here is one of this section's  key results.  

\begin{theorem} \label{thm:other-evals}
Consider the matrix $A(w)$ of order $n$, and let $V$ be an orthogonal matrix that diagonalises $A(w)$, i.e. $V^T A(w) V = {\rm{diag}}(\lambda_1, \ldots, \lambda_n)$.  For $3 \le j \le n$, there exists $\theta_j$ such that $\frac{(j - 2) \pi}{n - 1} \le \theta_j \le \frac{j \pi}{n + 1}$ and $\lambda_j = 2 \cos (\theta_j)$.  Moreover, for $1 \le j \le n$, if $\lambda_j = 2 \cos (\theta_j)$, then we have
\[
v_{1j}^2 = \begin{cases}
\frac{\displaystyle 2 \cos^2 \left( \frac{n - 1}{2} \theta_j \right)}{\displaystyle n + \frac{\sin (n \theta_j)}{\sin (\theta_j)}}, & j \text{ odd}; \\ ~ \\
\frac{\displaystyle 2 \sin^2 \left( \frac{n - 1}{2} \theta_j \right)}{\displaystyle n - \frac{\sin (n \theta_j)}{\sin (\theta_j)}}, & j \text{ even}.
\end{cases}
\]
\end{theorem}

\begin{proof} 
We give the proof in the case that $j$ is odd, but not for the case that $j$ is even as the argument is similar. 

Suppose $j$ is odd and consider the function $f(\theta) = w \cos \left( \left( \frac{n - 1}{2} \right) \theta \right) - \cos \left( \left( \frac{n + 1}{2} \right) \theta \right)$.  We have
\begin{align*}
f \left( \frac{(j - 2) \pi}{n - 1} \right) &= w \cos \left( \frac{(j - 2) \pi}{2} \right) - \cos \left( \frac{(n + 1) (j - 2) \pi}{2 (n - 1)} \right) \\ &= - \cos \left( \frac{(j - 2) \pi}{2} \right) \cos \left( \frac{(j - 2) \pi}{n - 1} \right) + \sin \left( \frac{(j - 2) \pi}{2} \right) \sin \left( \frac{(j - 2) \pi}{n - 1} \right), \\
&= (-1)^{\frac{j - 3}{2}} \sin \left( \frac{(j - 2) \pi}{n - 1} \right) \\
f \left( \frac{j \pi}{n + 1} \right) &= w \cos \left( \frac{(n - 1) j \pi}{2 (n + 1)} \right) - \cos \left( \frac{j \pi}{2} \right) \\
&= w \cos \left( \frac{j \pi}{2} \right) \cos \left( \frac{j \pi}{n + 1} \right) + w \sin \left( \frac{j \pi}{2} \right) \sin \left( \frac{j \pi}{n + 1} \right) \\
&= w (-1)^{\frac{j - 1}{2}} \sin \left( \frac{j \pi}{n + 1} \right).
\end{align*}
Since $f \left( \frac{(j - 2) \pi}{n - 1} \right)$ and $f \left( \frac{j \pi}{n + 1} \right)$ have opposite signs,  by the Intermediate Value Theorem, there exists $\theta_j$ such that $\frac{(j - 2) \pi}{n - 1} \le \theta_j \le \frac{j \pi}{n + 1}$ and $f(\theta_j) = 0$.  Hence, by Lemma~\ref{lem:evecs}, $2 \cos (\theta_j)$ is an eigenvalue of $A(w)$ with an  eigenvector $v$ such that $v_\ell = \cos ((\ell - \frac{n + 1}{2}) \theta), \ell=1, \ldots, n$. 

Next, we observe that 
\begin{align*}
v^T v &= \sum_{\ell = 1}^n \cos^2 \left( \left( \ell - \frac{n + 1}{2} \right) \theta_j \right) \\ 
&= \frac{1}{2} \sum_{\ell = 1}^n \cos \left( \left( 2 \ell - (n + 1) \right) \theta_j \right) + \frac{n}{2} \\
&= \frac{1}{4 \sin (\theta_j)} \sum_{\ell = 1}^n (\sin ((2 \ell - n) \theta_j) - \sin ((2 \ell - (n + 2)) \theta_j) ) + \frac{n}{2} \\
&= \frac{1}{2} \frac{\sin (n \theta_j)}{\sin (\theta_j)} + \frac{n}{2}
\end{align*}
from which the result follows.
\end{proof}

\begin{remark}\label{rmk:v-entry}
We claim that if $w>1,$ then $v_{1j}^2\le \frac{2}{n}+O(\frac{1}{n^2}), j=3, \ldots, n.$ In view of Theorem \ref{thm:other-evals}, it suffices to prove that $\big|\frac{\sin(n \theta_j)}{\sin(\theta_j)} \big|  $ is bounded for $j=3, \ldots, n.$ Fix such a $j,$ and for concreteness we suppose that $j$ is odd, so that $w \cos\left(\frac{n-1}{2} \theta_j \right) =\cos\left(\frac{n+1}{2} \theta_j \right). $ Using the sum of angles formula and the defining equation for $\theta_j,$ we have $\sin(n \theta_j) = \cos\left(\frac{n-1}{2} \theta_j \right) (w \sin\left(\frac{n-1}{2} \theta_j \right) + \sin\left(\frac{n+1}{2} \theta_j \right))$, 
$\sin( \theta_j) = \cos\left(\frac{n-1}{2} \theta_j \right) (-w \sin\left(\frac{n-1}{2} \theta_j \right) + \sin\left(\frac{n+1}{2} \theta_j \right)).$  

Consequently, we have 
$$\frac{\sin(n \theta_j)}{\sin(\theta_j)} = \frac{w \sin\left(\frac{n-1}{2} \theta_j \right) + \sin\left(\frac{n+1}{2} \theta_j \right)}{-w \sin\left(\frac{n-1}{2} \theta_j \right) + \sin\left(\frac{n+1}{2} \theta_j \right)}.$$ 
Let $x = \cos\left(\frac{n-1}{2} \theta_j \right)$ and note that $|wx|\le 1.$ Then $ \sin\left(\frac{n-1}{2} \theta_j \right) = \pm \sqrt{1-x^2}, \sin\left(\frac{n+1}{2} \theta_j \right) = \pm \sqrt{1-w^2x^2},$ and it follows that $\frac{\sin(n \theta_j)}{\sin(\theta_j)}$ is of the form $\frac{\pm w \sqrt{1-x^2}\pm \sqrt{1-w^2x^2}}{\mp w \sqrt{1-x^2}\pm \sqrt{1-w^2x^2}}.$  Furthermore, since $w>1,$ it follows that 
$|\mp w \sqrt{1-x^2}\pm \sqrt{1-w^2x^2}| \ge  w \sqrt{1-x^2} -  \sqrt{1-w^2x^2}.$ It is readily established $ w \sqrt{1-x^2} -  \sqrt{1-w^2x^2} \ge w-1$ whenever $|wx|\le 1.$ It now follows that $|\frac{\sin(n \theta_j)}{\sin(\theta_j)}| \le \frac{w+1}{w-1}.$ A similar argument applies when $j$ is even, and we now deduce that $v_{1j}^2\le \frac{2}{n}+O(\frac{1}{n^2}), j=3, \ldots, n,$ as claimed. 
\end{remark}

\section{Consequences for Fidelity}\label{subsec:fid} 

Our next goal is to develop a lower bound on the fidelity of state transfer from vertex $1$ to vertex $n$ when the Hamiltonian is $A(w)$. The fidelity at time $t$, denoted $p(t)$, is given by 
\[
p(t)= \left| \sum_{j=1}^n v_{1j}v_{nj}e^{it\lambda_j} \right|^2= \left|\sum_{j=1}^n v_{1j}^2(-1)^{j-1}e^{it\lambda_j} \right|^2.
\]
Applying the triangle inequality, we find that 
\[
\left|\sum_{j=1}^n v_{1j}^2(-1)^{j-1}e^{it\lambda_j}\right| \ge \left|v_{11}^2-v_{12}^2e^{it(\lambda_2-\lambda_1)}\right| -\left|\sum_{j=3}^n v_{1j}^2\right| = \left|v_{11}^2-v_{12}^2e^{it(\lambda_2-\lambda_1)}\right| +  v_{11}^2 + v_{12}^2-1.
\]
Evidently in the case that 
$|v_{11}^2-v_{12}^2e^{it(\lambda_2-\lambda_1)}| +  v_{11}^2 + v_{12}^2-1 \ge 0,$ we obtain $$p(t)\ge \left( |v_{11}^2-v_{12}^2e^{it(\lambda_2-\lambda_1)}| +  v_{11}^2 + v_{12}^2-1 \right)^2.$$

That observation prompts the next lemma. 

\begin{lemma}\label{lem:gxy}
Fix $\alpha \in [0,2\pi],$ and consider the function   $g(x,y)=\sqrt{x^2+y^2-2xy \cos \alpha} + x+y-1$ for $x,y>0.$ Fix $x_0, y_0>0,$ and suppose that $x \ge \tilde x, y \ge \tilde y.$ Then $$ g(x,y) \ge g(\tilde x, \tilde y) \ge g(x_0,y_0) + \frac{\partial g}{\partial x}\Big|_{(x_0,y_0)} (\tilde x-x_0) +\frac{\partial g}{\partial y}\Big|_{(x_0,y_0)} (\tilde y-y_0). $$
\end{lemma}

\begin{proof}
We have $\frac{\partial g}{\partial x} = \frac{x - y\cos \alpha}{\sqrt{x^2+y^2-2xy \cos \alpha}}+1 \ge 0$ 
and $\frac{\partial g}{\partial y} = \frac{y - x\cos \alpha}{\sqrt{x^2+y^2-2xy \cos \alpha}}+1 \ge 0,$ so that $g$ is increasing in both $x$ and $y$. In particular, $g(x,y) \ge g(\tilde x, \tilde y).$

A computation shows that the Hessian for $g$ is $$H = \frac{1-\cos^2\alpha}{(x^2+y^2-2xy \cos \alpha)^{\frac{3}{2}}} \left[\begin{array}{cc}y^2 &-xy\\-xy & x^2\end{array}\right]. $$  Evidently $H$ is positive semidefinite, and so we find from Taylor's theorem that 
\begin{equation*}
     g(\tilde x, \tilde y) \ge g(x_0,y_0) + \frac{\partial g}{\partial x}\Big|_{(x_0,y_0)} (\tilde x-x_0) +\frac{\partial g}{\partial y}\Big|_{(x_0,y_0)} (\tilde y-y_0). \qedhere
\end{equation*}
\end{proof}

We now derive a lower bound on the fidelity of state transfer. 

\begin{theorem}\label{thm:lowerbd} Suppose that $w+\frac{1}{w}-\tau_n \ge 2.$ Set $\alpha = t(\lambda_1-\lambda_2),$ and suppose that $\alpha$ is not an even multiple of $\pi$. Then 
\begin{align} \label{eqn:bd} \nonumber
&\quad |v_{11}^2-v_{12}^2e^{it(\lambda_2-\lambda_1)}| +  v_{11}^2 + v_{12}^2-1 \\ &\ge \frac{w^2-1}{w^2}\sqrt{\frac{1-\cos \alpha}{2}} -\frac{1}{w^2} + \frac{(w^2-1)(1-\cos^2 \alpha)}{\sqrt{8(1-\cos \alpha)}} \left(-\frac{8}{(n-1)w^2}\left[ \frac{1}{1+\frac{8}{(n-1)(w^2-1)}}\right] +\sigma_n \right), 
\end{align}
where 
\begin{align*}
\sigma_n
&= \frac{\tau_n}{(2w+\frac{16}{(n-1)(w-\frac{1}{w})})(w+\frac{1}{w}-\tau_n +\sqrt{(w+\frac{1}{w}-\tau_n)^2-4}+\frac{16}{(n-1)\sqrt{(w+\frac{1}{w}-\tau_n)^2-4}})} \times \\
&\quad \left[ 
w-\frac{1}{w} + \frac{(w+\frac{1}{w})(-2(w+\frac{1}{w})+\tau_n)}{w-\frac{1}{w} + \sqrt{(w+\frac{1}{w}-\tau_n)^2-4}}
+\frac{16(-2(w+\frac{1}{w})+\tau_n)}{(n-1)(w-\frac{1}{w})\sqrt{(w+\frac{1}{w}-\tau_n)^2-4}}
\right]. 
\end{align*}

In particular, if the right side of \eqref{eqn:bd} is nonnegative, then 
$$p(t) \ge \left( \frac{w^2-1}{w^2}\sqrt{\frac{1-\cos \alpha}{2}} -\frac{1}{w^2} + \frac{(w^2-1)(1-\cos^2 \alpha)}{\sqrt{8(1-\cos \alpha)}} \left(-\frac{8}{(n-1)w^2}\left[ \frac{1}{1+\frac{8}{(n-1)(w^2-1)}}\right] +\sigma_n \right)\right)^2.$$
\end{theorem}

\begin{proof}
We apply Lemma \ref{lem:gxy} with the parameter set $x=v_{11}^2$, $y=v_{12}^2$, $x_0=y_0=\frac{w^2-1}{2w^2}$, $\tilde x=b\left(w+\frac{1}{w}\right)$, and $\tilde y= b\left(w+\frac{1}{w}-\tau_n\right).$ 
Computations reveal with these parameters,  $\sqrt{x_0^2+y_0^2-2x_0y_0 \cos \alpha} = \frac{w^2-1}{w^2}\sqrt{\frac{1-\cos \alpha}{2}}$ and 
$\frac{\partial g}{\partial x}\Big|_{(x_0,y_0)} = \frac{\partial g}{\partial y}\Big|_{(x_0,y_0)} \frac{(w^2-1)(1-\cos^2 \alpha)}{\sqrt{8(1-\cos \alpha)}}. $ From Lemma \ref{lem:gxy}, we find that $$ g(x,y)\ge \frac{w^2-1}{w^2}\sqrt{\frac{1-\cos \alpha}{2}} -\frac{1}{w^2} + \frac{(w^2-1)(1-\cos^2 \alpha)}{\sqrt{8(1-\cos \alpha)}} \left( b\left(w+\frac{1}{w}\right) + b\left(w+\frac{1}{w}-\tau_n\right) -\frac{w^2-1}{w^2}   \right).$$ 
From Remark \ref{rmk:bfunction}, we find that 
\begin{align*}
&\quad b\left(w+\frac{1}{w}\right) + b\left(w+\frac{1}{w}-\tau_n\right) -\frac{w^2-1}{w^2} \\
&= -\frac{8}{(n-1)w^2}\left[ \frac{1}{1+\frac{8}{(n-1)(w^2-1)}}\right]\\ 
&\quad + \frac{\tau_n}{(2w+\frac{16}{(n-1)(w-\frac{1}{w})})(w+\frac{1}{w}-\tau_n +\sqrt{(w+\frac{1}{w}-\tau_n)^2-4}+\frac{16}{(n-1)\sqrt{(w+\frac{1}{w}-\tau_n)^2-4}})} \times \\
&\quad \left[ 
w-\frac{1}{w} + \frac{(w+\frac{1}{w})(-2(w+\frac{1}{w})+\tau_n)}{w-\frac{1}{w} + \sqrt{(w+\frac{1}{w}-\tau_n)^2-4}}
+\frac{16(-2(w+\frac{1}{w})+\tau_n)}{(n-1)(w-\frac{1}{w})\sqrt{(w+\frac{1}{w}-\tau_n)^2-4}}
\right].  
\end{align*}
The conclusion now follows. 
\end{proof}

\begin{cor} Suppose that $w \ge \sqrt{2}$, that $w+\frac{1}{w}-\tau_n \ge 2,$  and that $k \in \naturals.$ Then $p\left( \frac{(2k+1)\pi}{\lambda_1-\lambda_2}\right) \ge \left(\frac{w^2-2}{w^2}\right)^2.$
\end{cor}
\begin{proof}
Maintaining the notation of Theorem \ref{thm:lowerbd}, we have $\alpha = (2k+1)\pi,$ so that $\cos \alpha =-1.$ The inequality now follows from Theorem \ref{thm:lowerbd}. 
\end{proof}

\begin{remark}
In some of the results above, we have assumed the technical condition that $w+\frac{1}{w}-\tau_n \ge 2$. The following observations are straightforward. 
a) if $w+\frac{1}{w}-\tau_n \ge 2$ then $w+\frac{1}{w}-\tau_{n+k} \ge 2$ for any $k \in \naturals$; b) if $w>1, \exists n_0\in \naturals$ such that $w+\frac{1}{w}-\tau_n \ge 2$ for all $n \ge n_0$; 
c) $w+\frac{1}{w}-\tau_n \ge 2$ provided that $w$ is bounded below by the positive root of the polynomial $x^{2\lceil \frac{n}{2}\rceil +2} - x^{2\lceil \frac{n}{2}\rceil +1} -2x^2-x-1.$ For $n=4, \ldots, 8,$ the corresponding roots are approximately 
$1.6550 ,   1.4656 ,   1.4656 ,   1.3667, $ and $   1.3667,$ respectively. 
\end{remark}

Next we derive an upper bound on the fidelity. 
From the triangle inequality, we find that for any $t,$ 
$$p(t)\le \left(|v_{11}^2 -v_{12}^2e^{t(\lambda_2-\lambda_1)}| + \sum_{j=3}^n v_{1j}^2 \right)^2 = \left(|v_{11}^2 -v_{12}^2e^{t(\lambda_2-\lambda_1)}| +1-v_{11}^2 -v_{12}^2 \right)^2. $$ Letting $\alpha = t(\lambda_1-\lambda_2),$ we find readily that $p(t) \le \left(\sqrt{v_{11}^4+v_{12}^4-2v_{11}^2v_{12}^2\cos \alpha}+1-v_{11}^2-v_{12}^2 \right)^2.$ That observation prompts our interest in the function $h(x,y)=\sqrt{x^2 + y^2-2xy\cos \alpha}+1-x-y$, $x,y \in \mathbb{R},$ since evidently $p(t)\le h(v_{11}^2,v_{12}^2)^2.$ 

It is straightforward to determine that $h(x,y)$ is nonincreasing in both $x$ and $y$. Consequently, if we have $x\ge \tilde x>0, y\ge \tilde y>0,$ then $h(x,y)\le h(\tilde x,\tilde y).$ By Corollary \ref{cor:vbounds}, it follows that $h(v_{11}^2,v_{12}^2) \le h(b(\lambda_1), b(\lambda_2)).$ Recalling that $b(\lambda)$ is increasing for $\lambda>2,$ that $\lambda_1 > w+\frac{1}{w}, \lambda_2 \ge w+\frac{1}{w}-\tau_n,$ it now follows that if  $w+\frac{1}{w}-\tau_n>2,$ we have: 
\begin{align*}
 p(t) &\le  
 \left( \sqrt{b\left(w+\frac{1}{w}\right)^2 + b(w+\frac{1}{w}-\tau_n)^2 - 2 
b\left(w+\frac{1}{w}\right)b\left(w+\frac{1}{w}-\tau_n\right) \cos \alpha } \right. \\
&\quad \left. + 1-b\left(w+\frac{1}{w}\right)-b\left(w+\frac{1}{w}-\tau_n
\right)\right)^2.  
\end{align*}

We summarise the above discussion as follows. 

\begin{theorem}\label{thm:upperbdfid} Suppose that $n \in \naturals$ with $n \ge 2,$ and that $t>0.$ Set $\alpha =t(\lambda_1-\lambda_2),$ and suppose that $w+\frac{1}{w}-\tau_n>2.$ Then 
\begin{align*}
 p(t) &\le  
 \left( \sqrt{b\left(w+\frac{1}{w}\right)^2 + b(w+\frac{1}{w}-\tau_n)^2 - 2 
b\left(w+\frac{1}{w}\right)b\left(w+\frac{1}{w}-\tau_n\right) \cos \alpha } \right. \\
&\quad \left. + 1-b\left(w+\frac{1}{w}\right)-b\left(w+\frac{1}{w}-\tau_n
\right)\right)^2.  
\end{align*}
\end{theorem}

\begin{remark}
Recalling that $b(w+\frac{1}{w}), b(w+\frac{1}{w}-\tau_n) = \frac{w^2-1}{2w^2} + O\left(\frac{1}{n}\right),$ we find readily that 
$$p(t) \le \left(\frac{w^2-1}{w^2} \sqrt{\frac{1-\cos \alpha}{2}}  +\frac{1}{w^2} \right)^2+ O\left(\frac{1}{n}\right)= \left(\frac{w^2-1}{w^2} \left|\sin \frac{\alpha}{2}\right| +\frac{1}{w^2} \right)^2+ O\left(\frac{1}{n}\right) .$$ 
\end{remark}

\begin{cor}\label{cor:upper} 
Suppose that the hypothesis of Theorem \ref{thm:upperbdfid} holds. 
Suppose further that $p(t)=\gamma^2$ for some  $\gamma \ge 1-b(w+\frac{1}{w})- b(w+\frac{1}{w}-\tau_n) .$ Then necessarily $$\cos \alpha \le -1 + (1-\gamma)\left( \frac{b\left(w+\frac{1}{w}\right)+ b\left(w+\frac{1}{w}-\tau_n\right)}{b\left(w+\frac{1}{w}\right)b\left(w+\frac{1}{w}-\tau_n\right)}\right) -\frac{(1-\gamma)^2}{2b\left(w+\frac{1}{w}\right)b\left(w+\frac{1}{w}-\tau_n\right)}. $$ 

Let $\alpha = (2k+1)\pi +z$ where $k \in \naturals$, $|z|\le \pi.$ If $ 1 - (1-\gamma)\left( \frac{b\left(w+\frac{1}{w}\right)+ b\left(w+\frac{1}{w}-\tau_n\right)}{b\left(w+\frac{1}{w}\right)b\left(w+\frac{1}{w}-\tau_n\right)}\right) +\frac{(1-\gamma)^2}{2b\left(w+\frac{1}{w}\right)b\left(w+\frac{1}{w}-\tau_n\right)} \ge 0$ then 
\begin{equation*}
|z| \le \cos^{-1}\left( 1 - (1-\gamma)\left( \frac{b\left(w+\frac{1}{w}\right)+ b\left(w+\frac{1}{w}-\tau_n\right)}{b\left(w+\frac{1}{w}\right)b\left(w+\frac{1}{w}-\tau_n\right)}\right) +\frac{(1-\gamma)^2}{2b\left(w+\frac{1}{w}\right)b\left(w+\frac{1}{w}-\tau_n\right)}\right). 
\end{equation*}
\end{cor}

\begin{proof} 
From Theorem \ref{thm:upperbdfid}, we find that 
$$\sqrt{b\left(w+\frac{1}{w}\right)^2 + b(w+\frac{1}{w}-\tau_n)^2 - 2 
b\left(w+\frac{1}{w}\right)b\left(w+\frac{1}{w}-\tau_n\right) \cos \alpha }\ge \gamma -1 + b(w+\frac{1}{w})+ b(w+\frac{1}{w}-\tau_n).$$
Simplifying that inequality now yields 
$$\cos \alpha \le -1 + (1-\gamma)\left( \frac{b\left(w+\frac{1}{w}\right)+ b\left(w+\frac{1}{w}-\tau_n\right)}{b\left(w+\frac{1}{w}\right)b\left(w+\frac{1}{w}-\tau_n\right)}\right) -\frac{(1-\gamma)^2}{2b\left(w+\frac{1}{w}\right)b\left(w+\frac{1}{w}-\tau_n\right)}. 
$$ 

For $\alpha = (2k+1)\pi +z$ with $k \in \naturals,$  $|z|\le \pi,$ we find that 
$$\cos z \ge 1 - (1-\gamma)\left( \frac{b\left(w+\frac{1}{w}\right)+ b\left(w+\frac{1}{w}-\tau_n\right)}{b\left(w+\frac{1}{w}\right)b\left(w+\frac{1}{w}-\tau_n\right)}\right) +\frac{(1-\gamma)^2}{2b\left(w+\frac{1}{w}\right)b\left(w+\frac{1}{w}-\tau_n\right)}.$$
In particular, since the right hand side above is nonnegative, we may conclude that 
\begin{equation*} |z| \le \cos^{-1}\left( 1 - (1-\gamma)\left( \frac{b\left(w+\frac{1}{w}\right)+ b\left(w+\frac{1}{w}-\tau_n\right)}{b\left(w+\frac{1}{w}\right)b\left(w+\frac{1}{w}-\tau_n\right)}\right) +\frac{(1-\gamma)^2}{2b\left(w+\frac{1}{w}\right)b\left(w+\frac{1}{w}-\tau_n\right)}\right). \qedhere \end{equation*}
\end{proof}

\begin{remark}
We note that the lower bound on $\gamma$ in the hypothesis of Corollary \ref{cor:upper} can be written as  
$1-b(w+\frac{1}{w})- b(w+\frac{1}{w}-\tau_n) =\frac{1}{w^2} + O\left(\frac{1}{n}\right) . $
\end{remark}

\begin{remark} Here we consider the hypothesis $ 1 - (1-\gamma)\left( \frac{b\left(w+\frac{1}{w}\right)+ b\left(w+\frac{1}{w}-\tau_n\right)}{b\left(w+\frac{1}{w}\right)b\left(w+\frac{1}{w}-\tau_n\right)}\right) +\frac{(1-\gamma)^2}{2b\left(w+\frac{1}{w}\right)b\left(w+\frac{1}{w}-\tau_n\right)} \ge 0$ in Corollary \ref{cor:upper}. Considering the left side as a quadratic in $\gamma$ and using our estimates on  
$b\left(w+\frac{1}{w}\right)$ and $b\left(w+\frac{1}{w}-\tau_n\right), $ we find that  the larger of the two roots of the quadratic is given by $\frac{w^2 +\sqrt{2}-1}{\sqrt{2} w^2}+O\left(\frac{1}{n}\right).$ Thus, for sufficiently large $n,$ the condition $\gamma > \frac{w^2 +\sqrt{2}-1}{\sqrt{2} w^2}$ is sufficient to ensure that the desired inequality in the hypothesis holds. 
\end{remark}

\begin{remark}\label{zbound} 
Observe that the upper bound on $|z|$ in Corollary \ref{cor:upper} can be written as $$\cos^{-1}\left(1 - (1-\gamma)\frac{4w^2}{w^2-1} +\frac{2w^4(1-\gamma)^2}{(w^2-1)^2}\right) + O\left(\frac{1}{n}\right).$$ In particular, if $\gamma$ is close to $1,$ then necessarily $\alpha $ is close to an odd multiple of $\pi. $
\end{remark}

The following is immediate from Corollary \ref{cor:upper}. 

\begin{theorem}
Suppose that $w$ is a value that yields PGST from one
end vertex to the other. Let $t_j$  be a sequence of readout times
such that $p(t_j) \rightarrow 1$ as $j \rightarrow \infty$.  Then there are sequences $k_j \in \naturals$ and $z_j \in [-\pi, \pi]$ such that 
$ t_j = \frac{ (2k_j+1)\pi+z_j}{\lambda_1-\lambda_2}, j \in \naturals$ and 
$z_j \rightarrow 0$ as $j \rightarrow \infty.$
\end{theorem}

\section{Time Sensitivity}\label{sec:time}

In this section, we address the sensitivity of the fidelity of state transfer between end vertices, with respect to the readout time. Here is this section's main result. 

\begin{theorem}
Let $A(w)$ be the adjacency matrix of the path with loops of weight $w$ on the end vertices.  The sensitivity of the fidelity of transfer between the end vertices with respect to readout time is bounded by
\begin{align*}
    \left| \frac{dp}{dt} \right| &\le 2 v_{11}^2 v_{12}^2 (\lambda_1 - \lambda_2) + 2 v_{11}^2 (\lambda_1 + 2) \sum_{\ell = 3}^n v_{1 \ell}^2 + 2 v_{12}^2 (\lambda_2 + 2) \sum_{\ell = 3}^n v_{1 \ell}^2 + 4 \left( \sum_{\ell = 3}^n v_{1 \ell}^2 \right)^2 
    \\ &= \frac{2w^4 + 4w^3 - 2}{w^5} + O \left( \frac{1}{n} \right).
\end{align*}
\end{theorem}

\begin{proof}
Let $U(t) = V e^{i t \Lambda} V^T$.  Then we have
\[
U(t + h) = V e^{i (t + h) \Lambda} V^T = U(t) + i h V \Lambda e^{i t \Lambda} V^T + O(h^2).
\]
Then considering the $(1, n)$ entry, we have
\begin{align*}
    u_{1 n}(t + h) &= e_1^T U(t + h) e_n \\
    &= e_1^T V e^{i t \Lambda} V^T e_n + i h e_1^T V \Lambda e^{i t \Lambda} V^T e_n + O(h^2) \\
    &= e_1^T V \left( \sum_{j = 1}^n e^{i t \lambda_j} e_j e_j^T \right) V^T e_n + i h e_1^T V \left( \sum_{j = 1}^n \lambda_j e^{i t \lambda_j} e_j e_j^T \right) V^T e_n +O(h^2) \\
    &= \sum_{j = 1}^n v_{1j}v_{nj} e^{i t \lambda_j}  + i h \sum_{j = 1}^n v_{1j}v_{nj} \lambda_j e^{i t \lambda_j}  + O(h^2) \\
    &= \sum_{j = 1}^n v_{1j}^2 (-1)^{j - 1} e^{i t \lambda_j} + i h \sum_{j = 1}^n v_{1j}^2 (-1)^{j - 1} \lambda_j e^{i t \lambda_j} +O(h^2) \\
    &= \sum_{j = 1}^n v_{1j}^2 (-1)^{j - 1} (\cos (t \lambda_j) - h \lambda_j \sin (t \lambda_j)) + i \sum_{j = 1}^n v_{1j}^2 (-1)^{j - 1} (\sin (t \lambda_j) + h \lambda_j \cos (t \lambda_j)) +O(h^2). 
\end{align*}
Therefore, for the fidelity we obtain
\begin{align*}
    |u_{1 n}(t + h)|^2 &= \left( \sum_{j = 1}^n v_{1j}^2 (-1)^{j - 1} \cos (t \lambda_j) \right)^2 + \left( \sum_{j = 1}^n v_{1j}^2 (-1)^{j - 1} \sin (t \lambda_j) \right)^2 \\ &\quad - 2 h \sum_{j = 1}^n v_{1j}^2 (-1)^{j - 1} \cos (t \lambda_j) \sum_{\ell = 1}^n v_{1\ell}^2 (-1)^{\ell - 1} \lambda_\ell \sin (t \lambda_\ell) \\ &\quad + 2 h \sum_{j = 1}^n v_{1j}^2 (-1)^{j - 1} \sin (t \lambda_j) \sum_{\ell = 1}^n v_{1\ell}^2 (-1)^{\ell - 1} \lambda_\ell \cos (t \lambda_\ell) + O(h^2).
\end{align*}
We then calculate the time derivative for the fidelity $p$ by
\begin{align*}
    \frac{dp}{dt} &= 2 \sum_{j = 1}^n \sum_{\ell = 1}^n v_{1j}^2 v_{1\ell}^2 (-1)^{j + \ell} \sin (t \lambda_j) \lambda_\ell \cos(t \lambda_\ell) - 2 \sum_{j = 1}^n \sum_{\ell = 1}^n v_{1j}^2 v_{1\ell}^2 (-1)^{j + \ell} \cos (t \lambda_j) \lambda_\ell \sin (t \lambda_\ell) \\
    &= 2 \sum_{j = 1}^n \sum_{\ell = 1}^n v_{1j}^2 v_{1\ell}^2 (-1)^{j + \ell} \lambda_\ell \sin (t (\lambda_j - \lambda_\ell)) \\
    &= 2 v_{11}^2 \left( v_{12}^2 (-1) \lambda_2 \sin (t (\lambda_1 - \lambda_2)) + \sum_{\ell = 3}^n v_{1\ell}^2 (-1)^{\ell + 1} \lambda_\ell \sin (t (\lambda_1 - \lambda_\ell)) \right) \\
    &\quad + 2 v_{12}^2 \left( v_{11}^2 (-1) \lambda_1 \sin(t (\lambda_2 - \lambda_1)) + \sum_{\ell = 3}^n v_{2\ell}^2 (-1)^\ell \lambda_\ell \sin (t (\lambda_2 - \lambda_\ell)) \right) \\
    &\quad + 2 \sum_{j = 3}^n v_{1j}^2 v_{11}^2 (-1)^{j + 1} \lambda_1 \sin (t ( \lambda_j - \lambda_1))
     + 2 \sum_{j = 3}^n v_{1j}^2 v_{12}^2 (-1)^j \lambda_2 \sin (t ( \lambda_j - \lambda_2)) \\
    &\quad + 2 \sum_{j = 3}^n \sum_{\ell = 3}^n v_{1j}^2 v_{1\ell}^2 (-1)^{j + \ell} \lambda_\ell \sin (t (\lambda_j - \lambda_\ell)) \\
    &= 2 v_{11}^2 v_{12}^2 (\lambda_1 - \lambda_2) \sin (t (\lambda_1 - \lambda_2))
     + 2 v_{11}^2 \sum_{\ell = 3}^n v_{1\ell}^2 (-1)^{\ell} (\lambda_1 - \lambda_\ell) \sin (t (\lambda_1 - \lambda_\ell)) \\
    &\quad + 2 v_{12}^2 \sum_{\ell = 3}^n v_{1 \ell}^2 (-1)^{\ell + 1} (\lambda_2 - \lambda_\ell) \sin (t (\lambda_2 - \lambda_\ell))
     + 2 \sum_{j = 3}^n \sum_{\ell = 3}^n v_{1j}^2 v_{1 \ell}^2 (-1)^{j + \ell} \lambda_\ell \sin (t (\lambda_j - \lambda_\ell)).
\end{align*}
Therefore, applying the bounds $|\sin x| \le 1$ and $|\lambda_\ell| \le 2$ for $3 \le \ell \le n$, we obtain
\begin{align*}
    \left| \frac{dp}{dt} \right| &\le 2 v_{11}^2 v_{12}^2 (\lambda_1 - \lambda_2) + 2 v_{11}^2 (\lambda_1 + 2) \sum_{\ell = 3}^n v_{1 \ell }^2 + 2 v_{12}^2 (\lambda_2 + 2) \sum_{\ell = 3}^n v_{1 \ell }^2 + 4 \left( \sum_{\ell = 3}^n v_{1 \ell }^2 \right)^2.
\end{align*}
Now, $v_{11}^2, v_{12}^2 = \frac{w^2 - 1}{2w^2} + O \left( \frac{1}{n} \right)$ by Theorem~\ref{thm:eveclim}, $\lambda_1 - \lambda_2 = O \left( \frac{1}{n} \right)$ by Theorem~\ref{thm:even}, and $\lambda_1, \lambda_2 = w + \frac{1}{w} +O \left( \frac{1}{n} \right)$ by Theorem~\ref{thm:bds}, so the bound is then  
\[
    4 \left( \frac{w^2 - 1}{2w^2} \right) \left( \frac{(w + 1)^2}{w} \right) \frac{1}{w^2} + \frac{4}{w^4} + O \left( \frac{1}{n} \right) = \frac{2w^4 + 4w^3 - 2}{w^5} + O \left( \frac{1}{n} \right). \qedhere
\]
\end{proof}

Figure \ref{fig:time_sens} plots the (numerically computed)  time sensitivity and  fidelity as functions of time for the values $n=5, 10$ and $w=2, 4$. We note that for $w=2$ we have 
$\frac{2w^4 + 4w^3 - 2}{w^5} = 1.9375,$ while for $w=4$, we have $\frac{2w^4 + 4w^3 - 2}{w^5} =0.748046875.$ Inspecting the maximum and minimum values of the sensitivity in the plots, it appears that the rough bound of 
$\frac{2w^4 + 4w^3 - 2}{w^5} + O \left( \frac{1}{n} \right)$ on $\left| \frac{dp}{dt} \right| $ is accurate up to  a constant factor. 

\begin{figure}[htbp]
\centering
	\includegraphics[scale=.55]{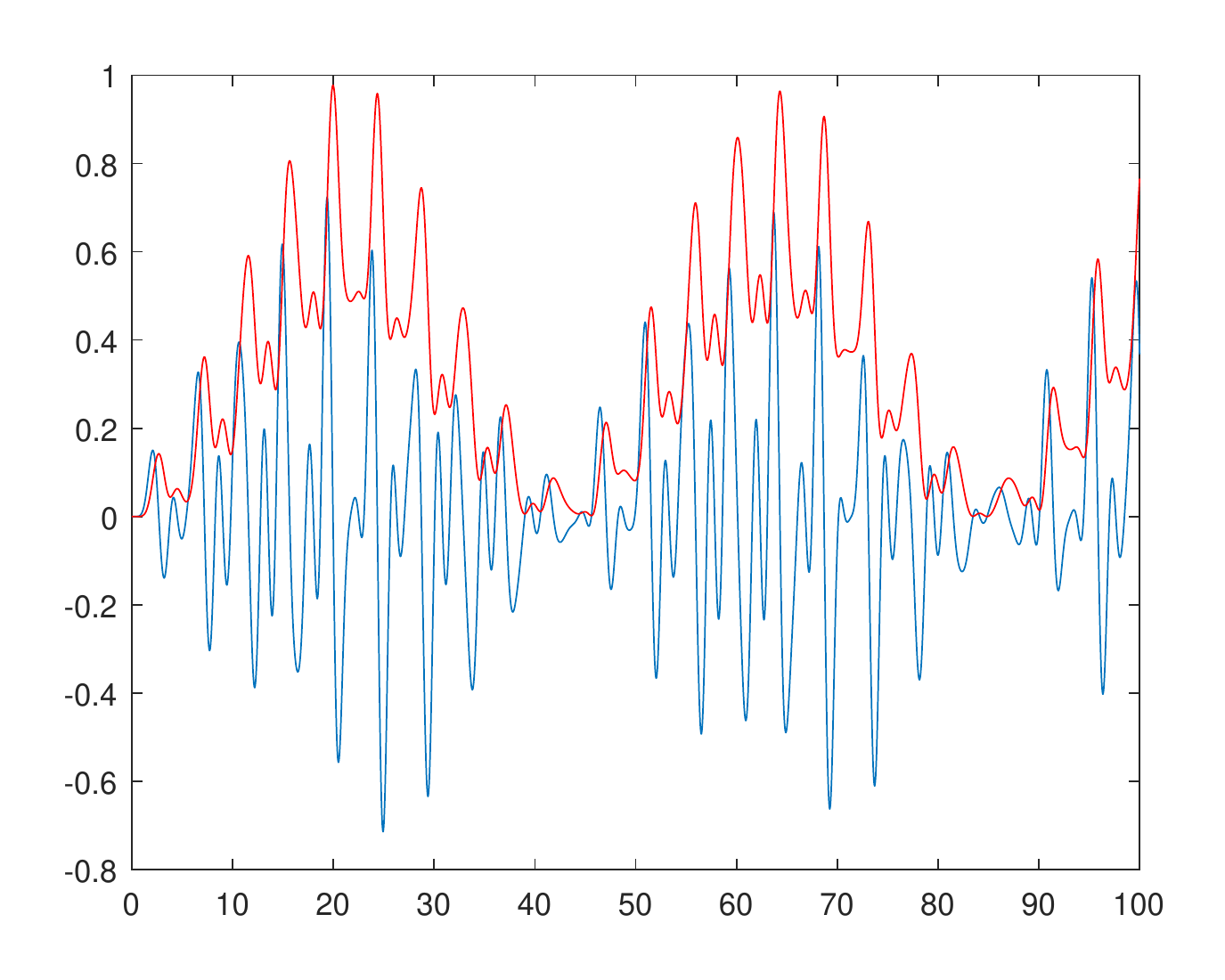}
	\includegraphics[scale=.55]{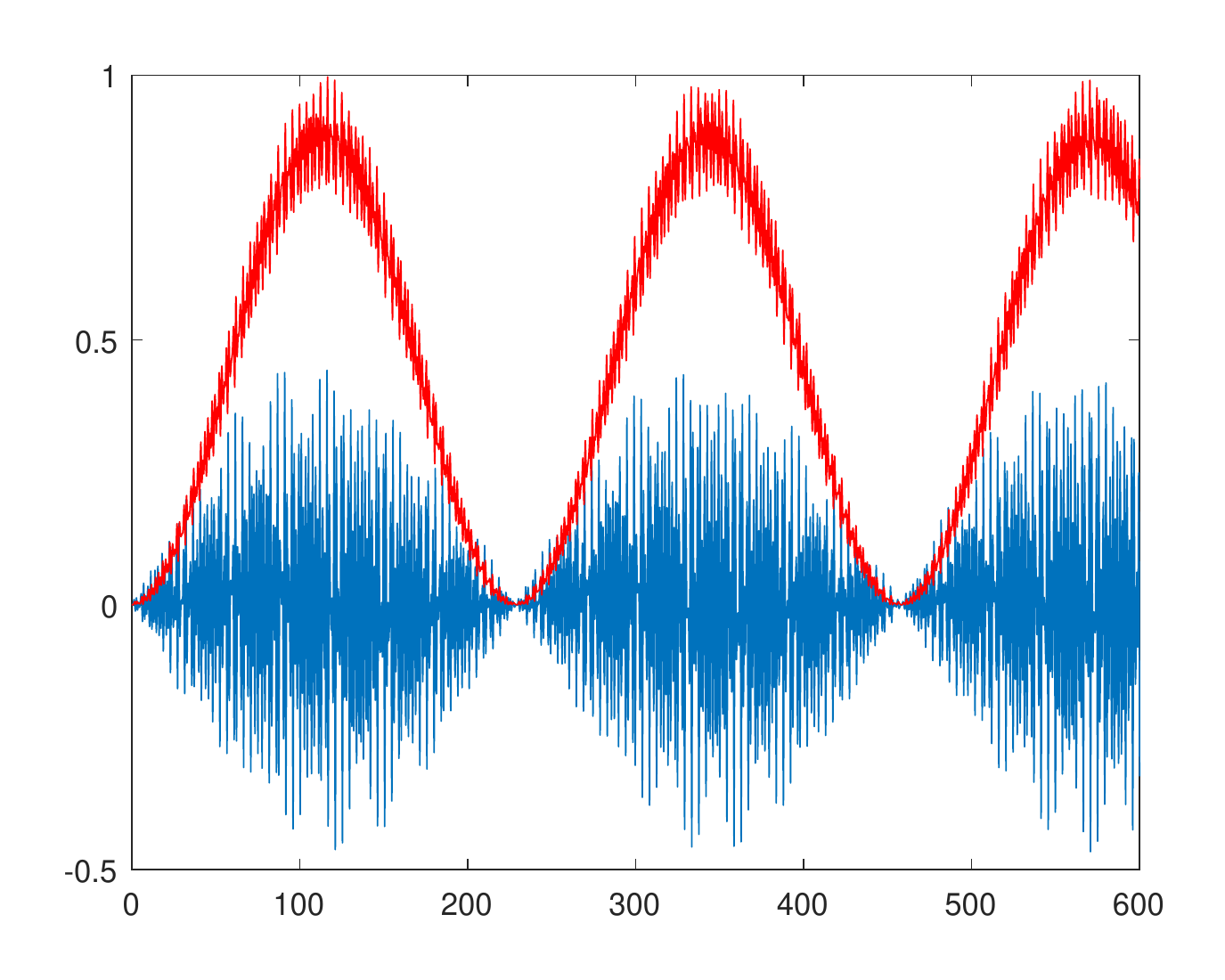}
	\includegraphics[scale=.55]{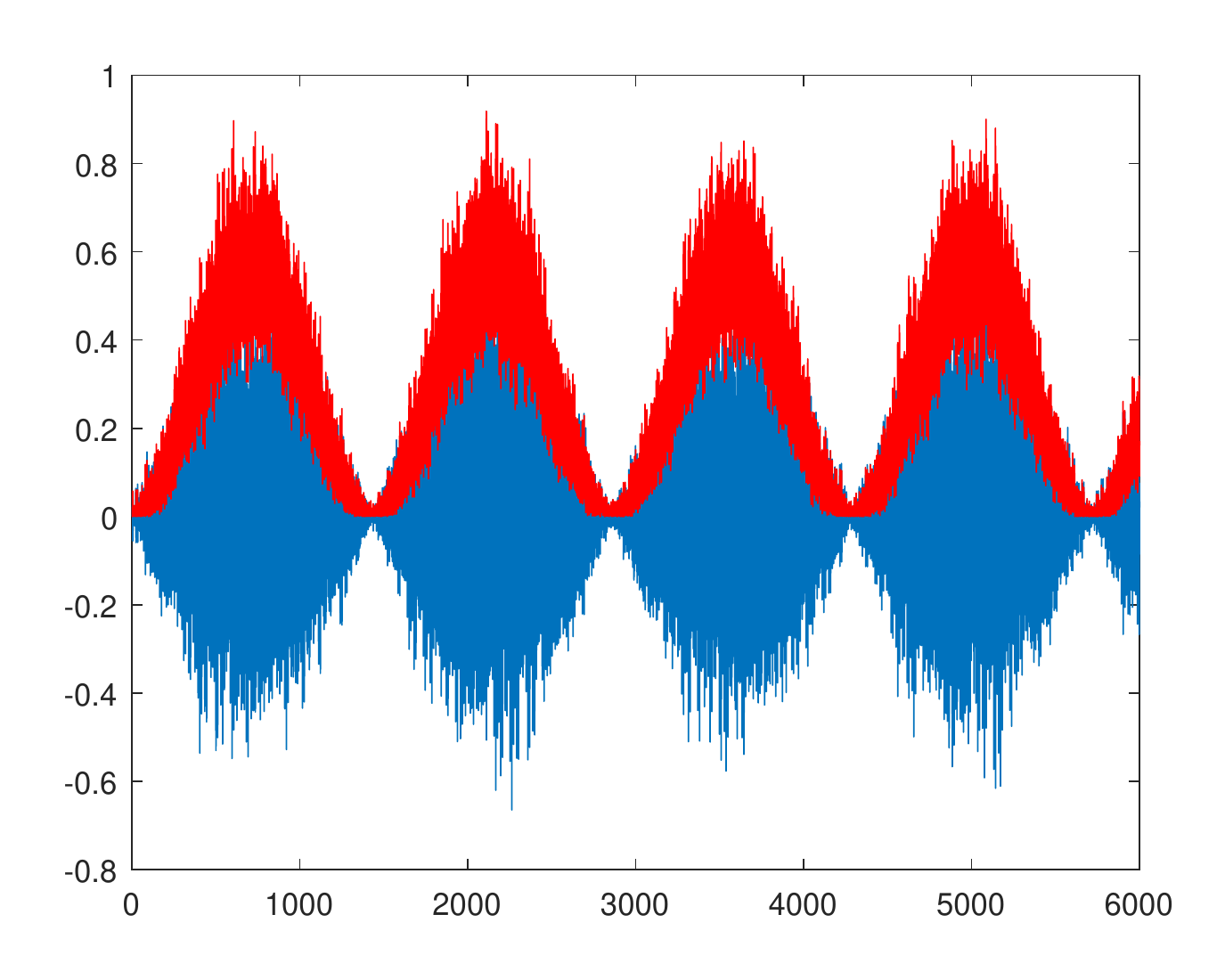}
	\includegraphics[scale=.55]{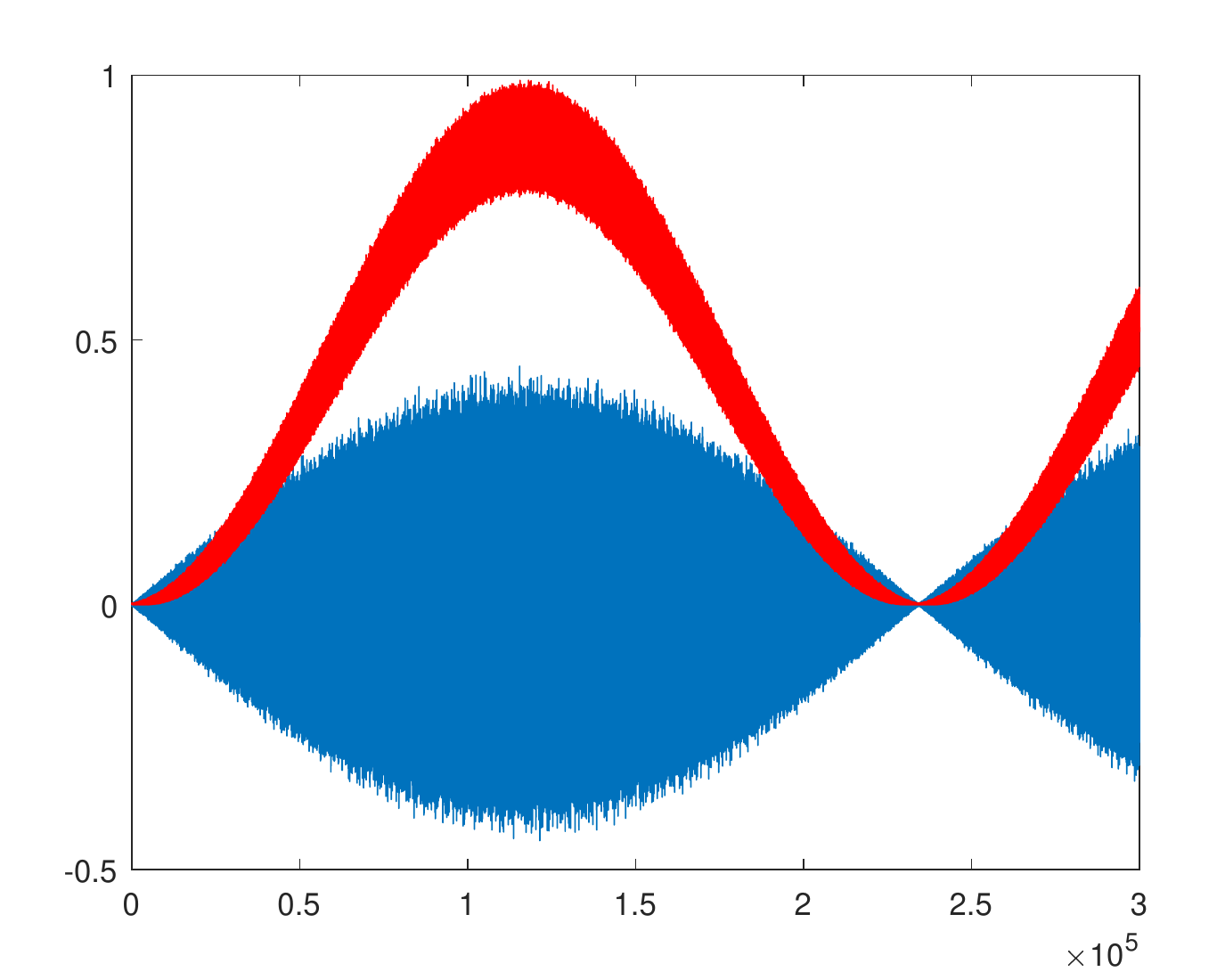}
	\caption{Graphs of time sensitivity (blue) and fidelity (red) against time; top row is $n=5, w=2, 4$, bottom row is $n=10, w=2,4$}\label{fig:time_sens}
\end{figure}

\section{Weight Sensitivity}\label{sec:weight}

Parallel to the result in Section \ref{sec:time},  we now estimate the sensitivity of the fidelity of state transfer between end vertices, with respect to the weight $w$. Here is our main result on this topic. 

\begin{theorem}
Let $A(w)$ be the adjacency matrix of the path with loops of weight $w$ on the end vertices.  The sensitivity of the fidelity of transfer between the end vertices with respect to loop weight is  bounded as follows: 
\[
\left| \frac{dp}{dw} \right| \le \frac{8 n^3}{3 \pi^2} +O(n^2) + 4t \left(  \frac{w^4-2w^2+3}{2w^6} +O \left( \frac{1}{n} \right) \right).
\]
\end{theorem}

\begin{proof}
The fidelity $p$ is given by
\[
p(t) = \left| \sum_{j = 1}^n (-1)^{j - 1} v_{1j}^2 e^{it \lambda_j} \right|^2 = \left( \sum_{j = 1}^n (-1)^{j - 1} v_{1j}^2 \cos (t \lambda_j) \right)^2 + \left( \sum_{j = 1}^n (-1)^{j - 1} v_{1j}^2 \sin (t \lambda_j) \right)^2.
\]
We calculate the derivative with respect to the  loop weight for the fidelity and obtain the following expression in terms of the derivatives of the eigenvalues and eigenvectors: 
\begin{equation*}
    \frac{dp}{dw} = 4 \sum_{k = 1}^n \sum_{j = 1}^n (-1)^{k + j} v_{1j}^2 v_{1k} \cos(t (\lambda_j - \lambda_k)) \frac{dv_{1k}}{dw} + 2t \sum_{k = 1}^n \sum_{j = 1}^n (-1)^{k + j} v_{1j}^2 v_{1k}^2 \sin(t (\lambda_j - \lambda_k)) \frac{d \lambda_k}{dw}.
\end{equation*}

Now, as calculated by Kirkland~\cite{K15}, the derivatives of the eigenvalues and eigenvectors are given by
\begin{align*}
    \frac{d \lambda_k}{dw} &= 2 v_{1k}^2, \\
    \frac{d v_{1k}}{dw} &= (\lambda_k I - A)^\dagger_{11} v_{1k} + (\lambda_k I - A)^\dagger_{1n} v_{nk} = \sum_{\ell \neq k} \frac{1}{\lambda_k - \lambda_\ell} v_{1 \ell}^2 v_{1 k} (1 + (-1)^{k + \ell}),
\end{align*}
where $M^\dagger$ denotes the Moore-Penrose inverse of $M$.

Thus, we have
\begin{align}\label{eq:dpdw} \nonumber
    \frac{dp}{dw} &= 4 \sum_{k = 1}^n \sum_{j = 1}^n (-1)^{k + j} v_{1j}^2 v_{1k} \cos (t (\lambda_j - \lambda_k)) \sum_{\ell \neq k} \frac{1}{\lambda_k - \lambda_\ell} v_{1 \ell}^2 v_{1k} (1 + (-1)^{k + \ell}) \\ &\quad + 4t \sum_{k = 1}^n \sum_{j = 1}^n (-1)^{k + j} v_{1j}^2 v_{1k}^4 \sin(t (\lambda_j - \lambda_k)).
\end{align}

We now bound the sensitivity in absolute value.  Considering the first term in \eqref{eq:dpdw}, we note $\sum v_{1j}^2 = 1$ and the cosine function is bounded by 1.  Splitting off the terms corresponding to the first and second eigenvalues and applying the triangle inequality, we obtain

\begin{align*}
    &\quad \left| 4 \sum_{k = 1}^n \sum_{j = 1}^n (-1)^{k + j} v_{1j}^2 v_{1k} \cos (t (\lambda_j - \lambda_k)) \sum_{\ell \neq k} \frac{1}{\lambda_k - \lambda_\ell} v_{1 \ell}^2 v_{1k} (1 + (-1)^{k + \ell}) \right| \\
    &\le 8  \sum_{m \neq 1} \left| \frac{1}{\lambda_1 - \lambda_m} v_{11}^2 v_{1m}^2 (1 + (-1)^{m + 1}) \right| + 8  \sum_{m \neq 2} \left| \frac{1}{\lambda_2 - \lambda_m} v_{12}^2 v_{1m}^2 (1 + (-1)^m) \right| \\ &\quad + 8 \sum_{k = 3}^n \sum_{\ell = k + 2}^n \left|  \frac{1}{\lambda_k - \lambda_\ell} v_{1k}^2 v_{1 \ell}^2 (1 + (-1)^{k + \ell}) \right|. \\
\end{align*}

Recall that $\lambda_1, \lambda_2 = w + \frac{1}{w} +O\left(\frac{1}{n}\right)$, $\lambda_m \le 2$ for $m \ge 3$, and $v_{11}^2, v_{12}^2 = \frac{w^2 - 1}{2w^2} +O\left(\frac{1}{n}\right)$; these yield 
\begin{eqnarray*}
&& 8  \sum_{m \neq 1} \left| \frac{1}{\lambda_1 - \lambda_m} v_{11}^2 v_{1m}^2 (1 + (-1)^{m + 1}) \right| + 8  \sum_{m \neq 2} \left| \frac{1}{\lambda_2 - \lambda_m} v_{12}^2 v_{1m}^2 (1 + (-1)^m) \right|\\
&& \le 16 \sum_{m = 3}^n \left( \frac{1}{w + \frac{1}{w} - 2} \left(\frac{w^2 - 1}{2 w^2} \right) +O \left( \frac{1}{n} \right) \right) v_{1 m}^2 \\
&&= O(1).
\end{eqnarray*}
Returning to the term corresponding to the remaining eigenvalues, we have that $\lambda_k - \lambda_\ell \ge \lambda_k - \lambda_{k + 2}$ for $\ell \ge k + 2;$ also, from Remark \ref{rmk:v-entry},  $v_{1j}^2 \le \frac{2}{n}+ O( \frac{1}{n^2}), j=3, \ldots, n.$ Hence we obtain
\begin{align*}
8 \sum_{k = 3}^n \sum_{\ell = k + 2}^n \left| \frac{1}{\lambda_k - \lambda_\ell} v_{1k}^2 v_{1 \ell}^2 (1 + (-1)^{k + \ell}) \right| &\le 8 \sum_{k = 3}^n \frac{n}{\lambda_k - \lambda_{k + 2}} \left( \frac{8}{n^2} +O \left(\frac{1}{n^3} \right) \right) \\ &=\frac{64}{n} \sum_{k = 3}^n \frac{1}{\lambda_k - \lambda_{k + 2}}  +O \left( \frac{1}{n} \right).  
\end{align*}
By Theorem~\ref{thm:other-evals}, we have $2 \cos \left( \frac{k \pi}{n + 1} \right) \le \lambda_k \le 2 \cos \left( \frac{(k - 2) \pi}{n - 1} \right)$, and we obtain
\begin{align*}
    \lambda_k - \lambda_{k + 2} &\ge 2 \cos \left( \frac{k \pi}{n + 1} \right) - 2 \cos \left( \frac{k \pi}{n - 1} \right). 
\end{align*}
Observe that $\cos \left( \frac{k \pi}{n + 1} \right) -  \cos \left( \frac{k \pi}{n - 1} \right) = \cos \left( \frac{kn \pi}{n^2 - 1} -\frac{k \pi}{n^2 - 1} \right) -  \cos \left( \frac{kn\pi}{n^2 - 1} + \frac{k \pi}{n^2 - 1} \right) = 2 \sin \left(\frac{kn\pi}{n^2-1}\right) \sin \left(\frac{k\pi}{n^2-1}\right).$ As $3 \le k \le n-2,$ it follows that
\begin{align*}
&\sin \left(\frac{k\pi}{n^2-1}\right) \ge  \sin \left(\frac{3\pi}{n^2-1}\right) = \frac{3 \pi}{n^2}+O \left( \frac{1}{n^3} \right),\text{ and} \\
&\sin \left(\frac{kn\pi}{n^2-1}\right) \ge \sin \left(\frac{(n-2)n\pi}{n^2-1}\right)= \sin \left(\frac{(2n-1)\pi}{n^2-1}\right) = \frac{2 \pi}{n}+O \left( \frac{1}{n^2} \right).
\end{align*}
It now follows that 
$  \lambda_k - \lambda_{k + 2} \ge \frac{24 \pi^2}{n^3}+O(\frac{1}{n^4}).$ In view of the observations above, we find that 
\begin{align*}
\left| 4 \sum_{k = 1}^n \sum_{j = 1}^n (-1)^{k + j} v_{1j}^2 v_{1k} \cos (t (\lambda_j - \lambda_k)) \sum_{\ell \neq k} \frac{1}{\lambda_k - \lambda_\ell} v_{1 \ell}^2 v_{1k} (1 + (-1)^{k + \ell}) \right| &\le \frac{64}{n} n \left( \frac{n^3}{24 \pi^2} + O(n^2) \right) \\ &= \frac{8 n^3}{3 \pi^2} +O(n^2).
\end{align*}
Therefore, the first term of \eqref{eq:dpdw} is bounded above in absolute value by $\frac{8 n^3}{3 \pi^2} +O(n^2).$

Now, considering the second term  of \eqref{eq:dpdw} for the sensitivity, we note that sine is an odd function, and splitting off the terms corresponding to the first and second eigenvalues, we obtain
\begin{align*}
    &\quad \sum_{k = 1}^n \sum_{j = 1}^n (-1)^{k + j} v_{1j}^2 v_{1k}^4 \sin(t (\lambda_j - \lambda_k)) \\
    &= (-1) \sin(t(\lambda_2 - \lambda_1)) (v_{11}^4 v_{12}^2 - v_{12}^4 v_{11}^2) + \sum_{j = 3}^n (-1)^{j + 1} \sin(t(\lambda_j - \lambda_1)) (v_{11}^4 v_{1j}^2 - v_{11}^2 v_{1j}^4) \\ 
    &\quad + \sum_{j = 3}^n (-1)^j \sin(t(\lambda_j - \lambda_2)) (v_{12}^4 v_{1j}^2 - v_{12}^2 v_{1j}^4) + \sum_{k = 3}^n \sum_{j = 3}^n (-1)^{k + j} v_{1j}^2 v_{1k}^4 \sin (t (\lambda_j - \lambda_k)).
\end{align*}
We have  $v_{11}^2, v_{12}^2 = \frac{w^2 - 1}{2w^2} +O\left(\frac{1}{n}\right)$, $v_{11}^2 - v_{12}^2 =O\left(\frac{1}{n}\right)$, and $\sum_{j = 3}^n v_{1j}^2 = \frac{1}{w^2}+O\left(\frac{1}{n}\right)$.  Applying these estimates, bounding sine by 1, and using the triangle inequality, we obtain 
\[
\left| 4t \sum_{k = 1}^n \sum_{j = 1}^n (-1)^{k + j} v_{1j}^2 v_{1k}^4 \sin(t (\lambda_j - \lambda_k)) \right| \le 4t \left(\left( \frac{w^2 - 1}{2w^2} \right)^2 \left(\frac{2}{w^2} \right) + \frac{1}{w^6} + O \left( \frac{1}{n} \right)\right).
\]
Combining these terms, we obtain the desired inequality.
\end{proof}

From the above, we deduce that as $n \rightarrow \infty,$ the absolute value of the coefficient of $t$ in our expression for $\frac{dp}{dw}$ is bounded above by a function that is asymptotic to 
$\frac{2(w^4-2w^2+3)}{w^6}.$ 

Figure \ref{fig:wt_sens} plots the (numerically computed)  weight sensitivity  as a function of time for the values $n=5, 10$ and $w=2, 4$. Observe that the values of $\left| \frac{dp}{dw}\right|$ at local maxima are increasing in the readout time, as anticipated by the results in this section. For  $w=2, 4$ we have $\frac{2(w^4-2w^2+3)}{w^6} \approx 0.3431, 0.1108,$ respectively.  Inspecting Figure \ref{fig:wt_sens}, we see that the behaviour of the local extrema of $ \frac{dp}{dw}$ is commensurate with those decimal values. 

\begin{figure}[htbp]
\centering
	\includegraphics[scale=.55]{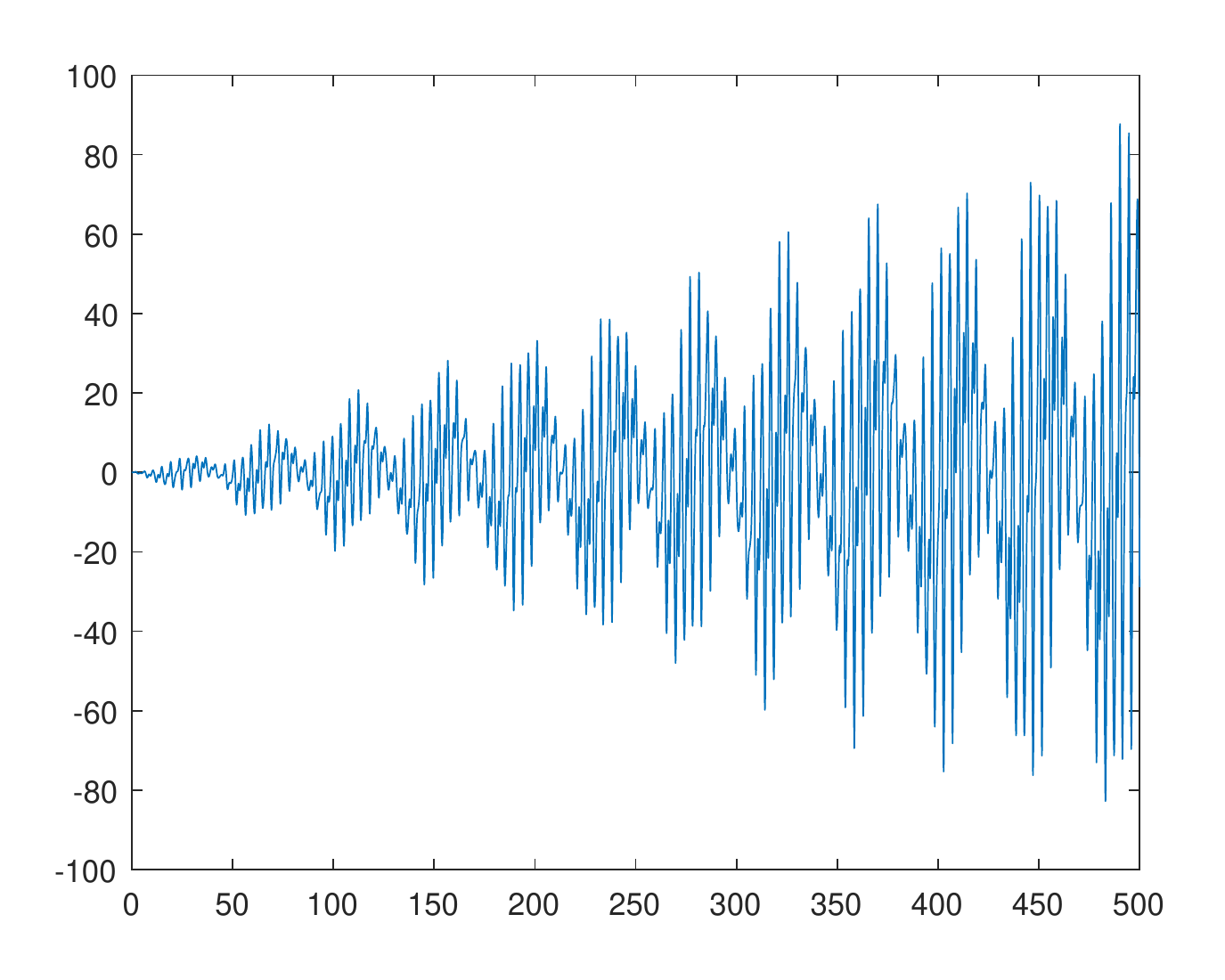}
	\includegraphics[scale=.55]{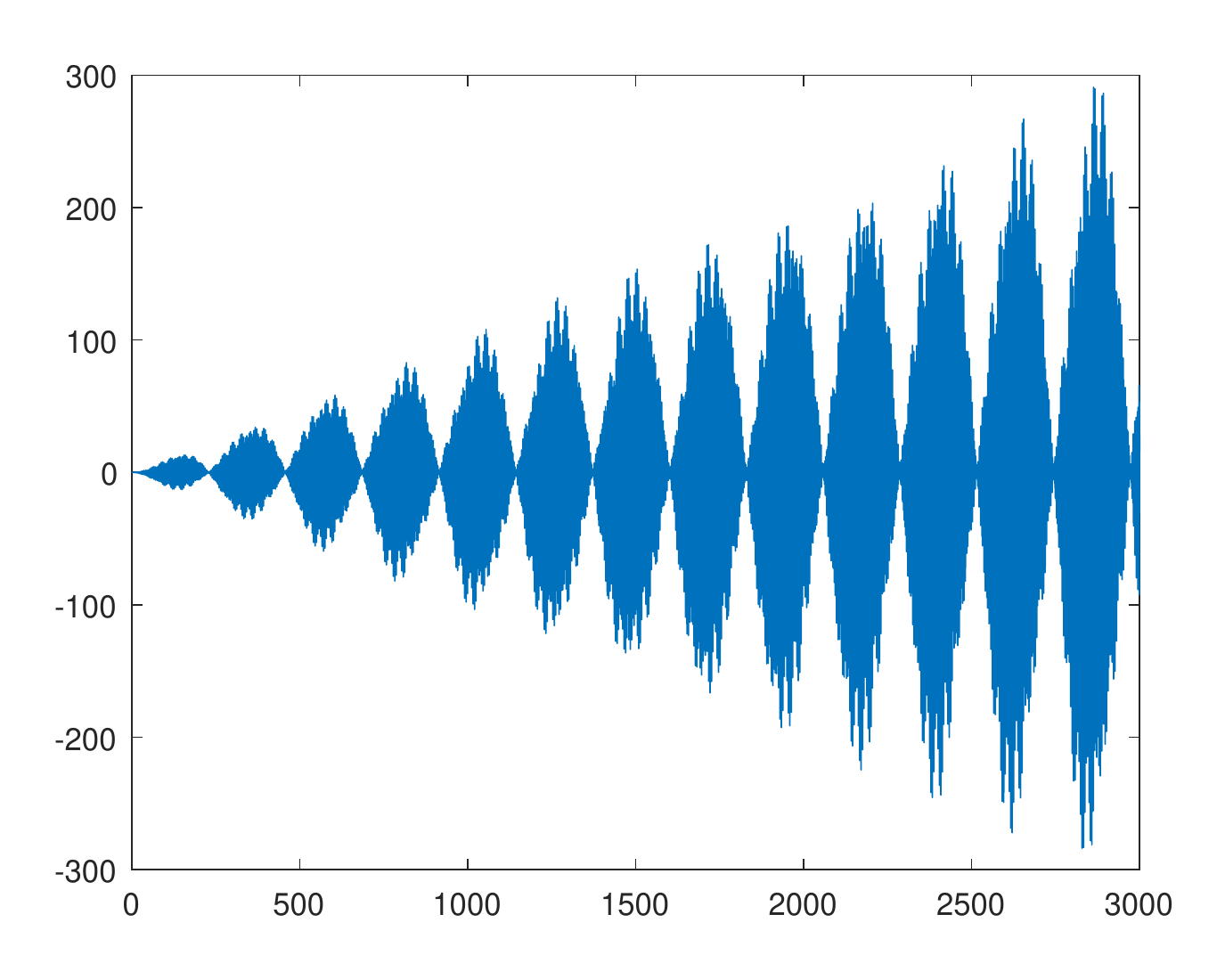}
	\includegraphics[scale=.55]{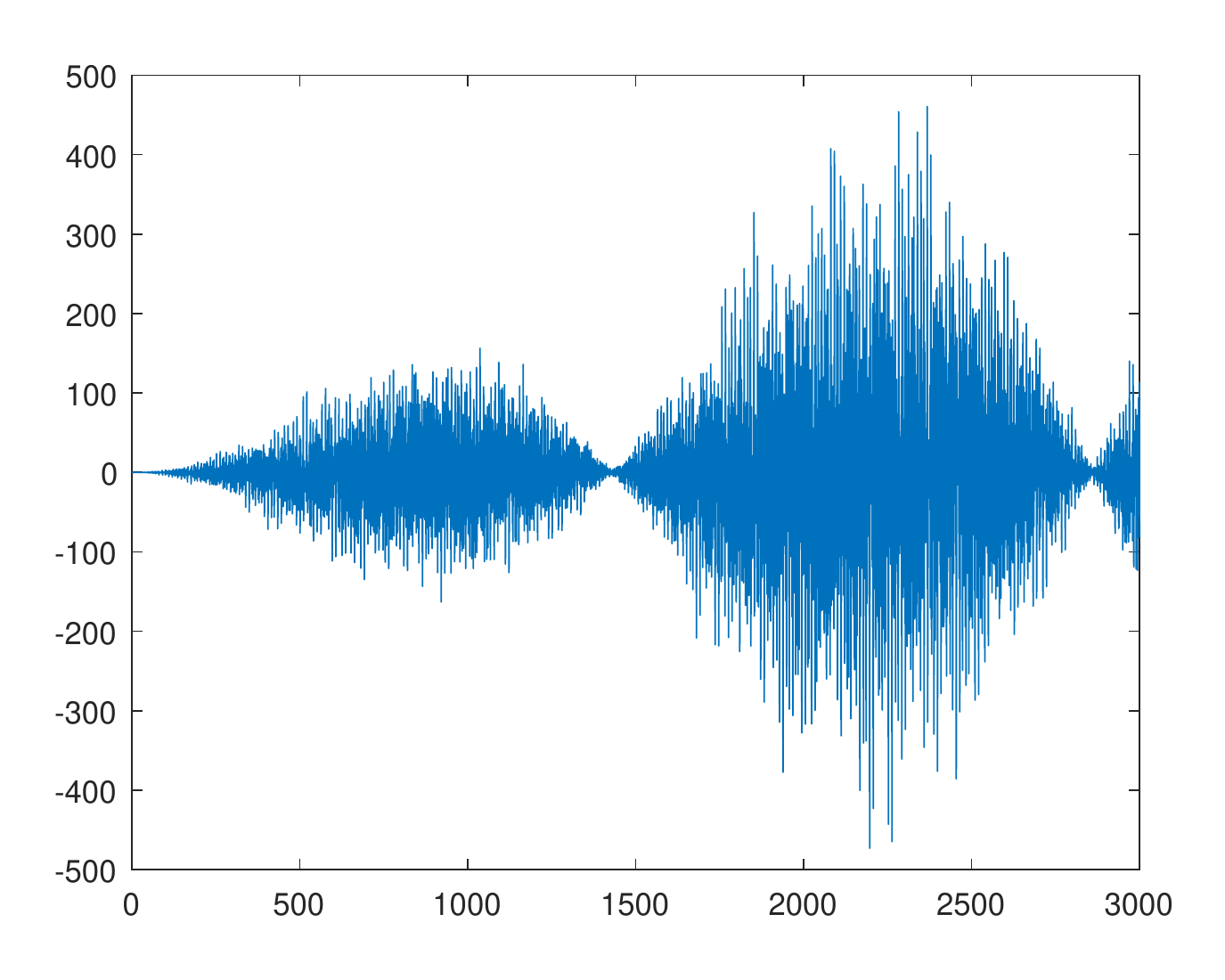}
	\includegraphics[scale=.55]{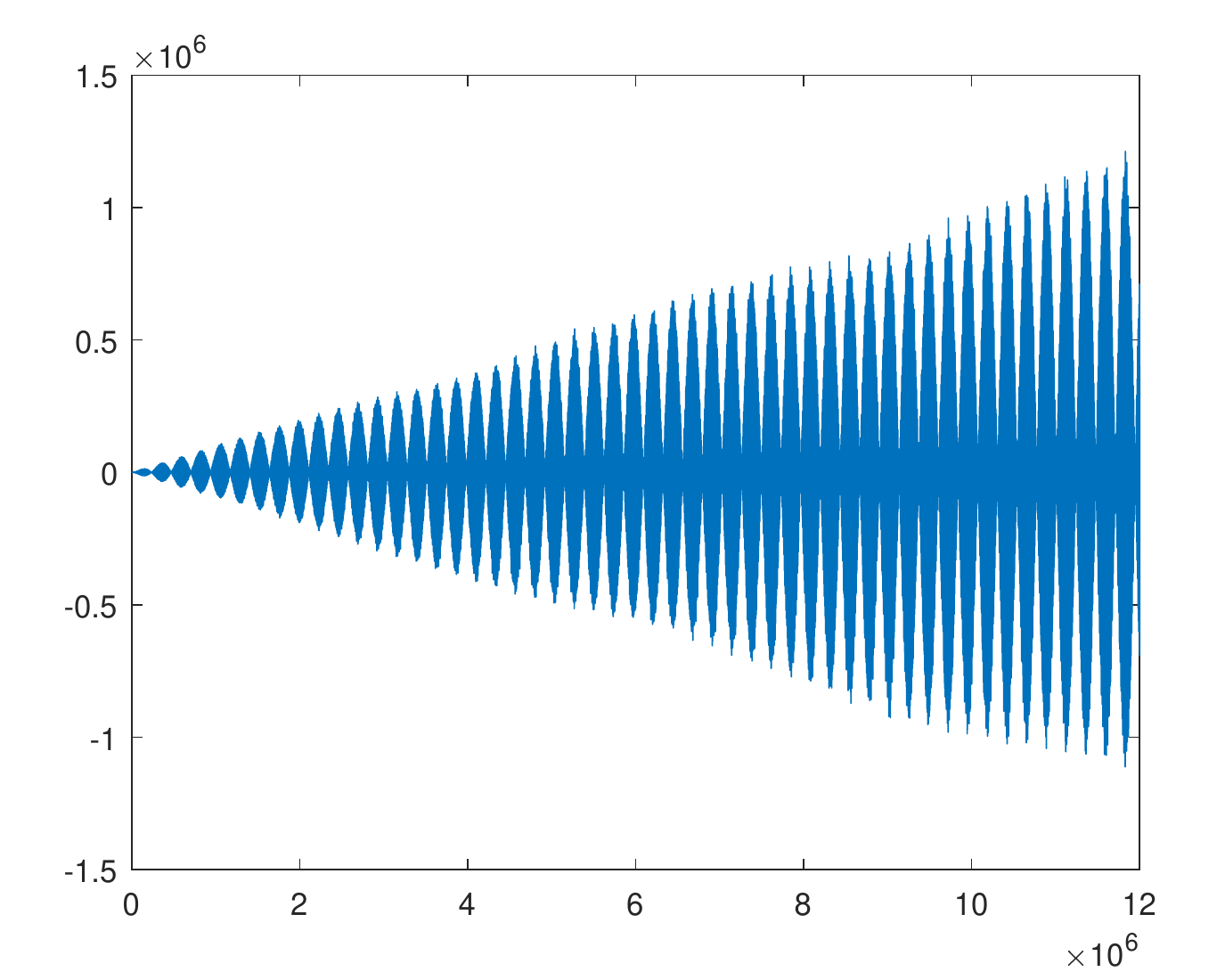}
	\caption{Graphs of weight sensitivity against time; top row is $n=5, w=2, 4$, bottom row is $n=10, w=2,4$}\label{fig:wt_sens}
\end{figure}


\begin{thebibliography}{10}
\bibitem{ACDE04} C. Albanese, M. Christandl, N. Datta, A Ekert.  Mirror inversion of quantum states in linear registers.  \emph{Phys. Rev. Lett.} 93, 230502 (2004).

\bibitem{BCGS17} L. Banchi, G. Coutinho, C. Godsil, S. Severini.  Pretty good state transfer in qubit chains--the Heisenberg Hamiltonian.  \emph{J. Math. Phys.} 58(3), 032202 (2017).

\bibitem{B03} S. Bose.  Quantum communication through an unmodulated spin chain.  \emph{Phys. Rev. Lett.} 91(20), 207901 (2003).

\bibitem{CLMS09} A. Casaccino, S. Lloyd, S. Mancini, S. Severini.  Quantum state transfer through a qubit network with energy shifts and fluctuations.  \emph{Int. J. Quantum Inf.} 7(8), 1417--1427 (2009).

\bibitem{CMF16} X. Chen, R. Mereau, D.L. Feder.  Asymptotically perfect efficient quantum state transfer across unifrom chains with two impurities.  \emph{Phys. Rev. A} 93, 012343 (2016).

\bibitem{CDEL04} M. Christandl, N. Datta, A. Ekert, A.J. Landahl.  Perfect state transfer in quantum spin networks.  \emph{Phys. Rev. Lett.} 92, 187902 (2004).

\bibitem{CDDEKL05} M. Christandl, N. Datta, T.C. Dorlas, A. Ekert, A. Kay, A.J. Landahl.  Perfect transfer of arbitrary states in quantum spin networks.  \emph{Phys. Rev. A} 71, 032312 (2005).

\bibitem{CGvB17} G. Coutinho, K. Guo, C.M. van Bommel.  Pretty good state transfer between internal nodes of path.  \emph{Quantum Inf. Comput.} 17(9-10), 825--830 (2017).

\bibitem{G12} C. Godsil.  State transfer on graphs.  \emph{Discrete Math.} 312(1), 129--147 (2012).

\bibitem{GKSS12} C. Godsil, S. Kirkland, S. Severini, J. Smith.  Number-theoretic nature of communication in quantum spin systems.  \emph{Phys. Rev. Lett.} 109(5), 050502 (2012).

\bibitem{KSA12} V. Karimipour, M. Sarmadi Rad, M. Asoudah.  Perfect quantum state transfer in two- and three dimensional structures.  \emph{Phys. Rev. A} 85, 010302 (2012).

\bibitem{K10} A. Kay.  Perfect, efficient, state transfer and its applications as a constructive tool.  \emph{Int. J. Quantum Inf.} 8(4), 641 (2010).

\bibitem{KLY17a} M. Kempton, G. Lippner, S.-T. Yau.  Perfect state transfer on graphs with a potential.  \emph{Quantum Inf. Comput.} 17(3), 303--327 (2017).

\bibitem{KLY} M. Kempton, G. Lippner, S.-T. Yau.  Pretty good quantum state transfer in symmetric spin networks via magnetic field.  \emph{Quantum Information Processing} 16:210 (2017).

\bibitem{K15} S. Kirkland.  Sensitivity analysis of perfect state transfer in quantum spin networks.  \emph{Linear Algebra and Its Applications} 472, 1--30 (2015).

\bibitem{LLY} Y. Lin, G. Lippner, S.-T. Yau. Quantum tunneling on graphs. \emph{Commun. Math. Phys.} 311, 113–132 (2012).

\bibitem{LSU12} T. Linneweber, J. Stolze, G.S. Uhrig.  Perfect state transfer in XX chains induced by boundary magnetic fields.  \emph{Int. J. Quantum Inf.} 10, 1250029 (2012).

\bibitem{LASP13} S. Lorenzo, T.J.G. Apollaro, A. Sindona, F. Plastina.  Quantum-state transfer via resonant tunneling through local-field-induced barriers.  \emph{Phys. Rev. A} 87, 042313 (2013).

\bibitem{PK11} P.J. Pemberton-Ross and A. Kay.  Perfect quantum routing in regular spin networks.  \emph{Phys. Rev. Lett.} 106, 020503 (2011).

\bibitem{S11} D. Stevanovic.  Applications of graph spectra in quantum physics.  \emph{Selected Topics in Applications of Graph Spectra} 85--111 (2011).

\bibitem{vB19} C.M. van Bommel.  A complete characterization of pretty good state transfer on paths.  \emph{Quantum Inf. Comput.} 19(7-8), 601--608 (2019).

\bibitem{VZ12} L. Vinet, A Zhedanov.  Almost perfect state transfer in quantum spin chains.  \emph{Phys. Rev. A} 86, 0523119 (2012).

\bibitem{WLKGGB05} A. W\'ojcik, T. {\L}uczak, P. Kurzy\'nski, A. Grudka, T. Gdala, M. Bednarska.  Unmodulated spin chains as universal quantum wires.  \emph{Phys. Rev. A} 72, 034303 (2005).
\end{thebibliography}
\end{document}